\documentclass[12pt]{article}
\usepackage[left=1in,top=1in,right=1in,bottom=1in,nohead]{geometry}

\usepackage{url,subfig,latexsym,amsmath,amssymb, amsfonts,enumerate,
  epsfig, graphics,times, amsthm,mathtools}
\usepackage{enumitem}
\usepackage{bbm}
\usepackage{color}
\usepackage{nameref}
\usepackage[colorlinks]{hyperref}
\usepackage[export]{adjustbox}
\usepackage{wrapfig}
\usepackage{float}
\usepackage[numbers, sort&compress]{natbib}
\newtheorem{thm}{Theorem}[section]
\newtheorem{cor}[thm]{Corollary}
\newtheorem{lem}[thm]{Lemma}
\newtheorem{prop}[thm]{Proposition}

\newtheorem*{thm*}{Theorem}

\usepackage{bbm}
\usepackage[sort&compress,numbers]{natbib}
\usepackage{tikz}      
\usetikzlibrary{shapes,automata,positioning,arrows,fit}
\usetikzlibrary{matrix}
\usetikzlibrary{positioning}
\usetikzlibrary{shapes.misc, positioning}
\usepackage{titlesec}

\newtheorem{defn}{Definition}[section]

\newtheorem{example}{Example}[section]

\theoremstyle{definition}
\newtheorem{rmk}[thm]{Remark}

\newcommand{\mybox}[1]{%
  \setbox0=\hbox{#1}%
  \setlength{\@tempdima}{\dimexpr\wd0+13pt}%
  \begin{tcolorbox}[colframe=mycolor,boxrule=0.5pt,arc=4pt,
      left=6pt,right=6pt,top=6pt,bottom=6pt,boxsep=0pt,width=\@tempdima]
    #1
  \end{tcolorbox}
}
\usepackage{xr}

\newcommand{\dlim}{\displaystyle \lim\limits}
\newcommand{\dsum}{\displaystyle \sum\limits}

\definecolor{cred}{rgb}{0.5, 0.0, 0.13}
\definecolor{cadmiumgreen}{rgb}{0.0, 0.42, 0.24}
\definecolor{bblue}{rgb}{0.0, 0.44, 1.0}
\definecolor{black}{rgb}{0.0, 0.0, 0.0}

\def\S{\mathcal S}
\def\C{\mathcal C}
\def\Re{\mathcal R}
\def\R{\mathbb R}

\def\Z{\mathbb Z}

\def\K{\mathcal K}
\def\Td1{T^{D,1}_{\{x_n\}}}
\def\Ts1{T^{S,1}_{\{x_n\}}}

\def\K{\mathcal{K}}

\def\Tid1{T^{D,1}_{\{x_n+\zeta_i\}}}
\def\Tis1{T^{S,1}_{\{x_n+\zeta_i\}}}
\def\Tjs1{T^{S,1}_{\{x_n+\zeta_j\}}}
\def\Tjd1{T^{D,1}_{\{x_n+\zeta_j\}}}

\newcount\colveccount
\newcommand*\colvec[1]{
        \global\colveccount#1
        \begin{pmatrix}
        \colvecnext
}
\def\colvecnext#1{
        #1
        \global\advance\colveccount-1
        \ifnum\colveccount>0
                \\
                \expandafter\colvecnext
        \else
                \end{pmatrix}
        \fi
}

\title{Absolutely Robust Controllers for Chemical Reaction Networks}

\author{Jinsu Kim \and
German Enciso }

\begin{document}

\maketitle

\begin{abstract}
\noindent 
In this work, we design a type of controller that consists of adding a specific set of reactions to an existing mass-action chemical reaction network in order to control a target species.  This set of reactions is effective for both deterministic and stochastic networks, in the latter case controlling the mean as well as the variance of the target species. We employ a type of network property called absolute concentration robustness (ACR). We provide applications to the control of a multisite phosphorylation model as well as a receptor-ligand signaling system.

For this framework, we use the so-called deficiency zero theorem from chemical reaction network theory as well as multiscaling model reduction methods.  We show that the target species has approximately Poisson distribution with the desired mean.  We further show that ACR controllers can bring robust perfect adaptation to a target species and are complementary to
a recently introduced antithetic feedback controller used for stochastic chemical reactions.

\end{abstract}


\section{Introduction}

In this paper we propose a set of synthetic controllers that can be added to a given chemical reaction network in order to control the concentration or copy number of a given species of interest.  Chemical reaction networks describe a variety of problems in engineering and biology, and there has recently been a surge in interest for stochastic models of such networks \cite{AndProdForm, enciso2019embracing, StoQSSA2019, Paulsson2004}.  Stochastic effects are important in order to describe the noise inherent in reactions with low numbers of molecules, as is often the case inside individual cells.  The techniques proposed in this paper will be shown to apply both in the deterministic case, where concentrations are described by ordinary differential equations, as well as in the stochastic case where the dynamics are described by a continuous time, discrete space Markov process.

The controllers used in our framework are inspired by a property called \textit{absolute concentration robustness} (ACR), which guarantees that a species has the same steady state value regardless of the initial conditions. We provide both a theoretical framework and computational simulations in several specific biochemical systems to show that an ACR controller can shift all positive steady state values of a target species towards a desired value. We also show that in stochastic networks satisfying certain topological criteria, the ACR controller can account for the intrinsic noise in the chemical reaction.  We approximate the behaviour of the target species using a reduced chemical reaction model derived through multiscaling analysis. Our stochastic analysis assumes certain conditions on the topology of the controlled network that are described using the so-called \emph{deficiency} and \emph{weak reversibility} of the system. These two theoretical tools will be combined to calculate the behaviour of the reduced system, as well as to show that the behaviour of the target species in the reduced system approximates that of the original network. Using computational simulations we also explore the \textit{robust perfect adaptation} of the target species in the controlled system, a highly desirable goal in control theory.

When a dynamical system has multistationarity or the dynamics are confined to a lower-dimensional subset by conservation relations among species, the long term behaviour depends on the initial conditions of the system. In general different initial concentrations may lead to different long term steady states of the different species.  However sometimes the positive steady state values of a species of interest are identical independent of the initial conditions. Such a system is said to possess the ACR property as the steady state of the dynamics is robust to the initial conditions. In that case, the species with identical steady state values is called the ACR species.  This counter-intuitive dynamical aspect was proposed by Shinar and Feinberg in 2010 \cite{shinar2010structural}, where they further provided network topological conditions ensuring that the associated deterministic system admits ACR. 

For a simple example of an ACR system consider the following network, which will be the basic ACR controller throughout this manuscript, 
\begin{align}\label{eq:ACRcontroller}
Z+A\xrightarrow{\theta} 2Z, \quad Z\xrightarrow{\mu} A.
\end{align}
Both reactions produce and consume the same amount of $A$ and $Z$, hence the total amount of $Z+A$ is conserved.  One can think of $A$ and $Z$ as being different forms of the same protein, say active and inactive.  Let $a(t)$ and $z(t)$ be the concentration of the species $A$ and $Z$, respectively.  Assuming the associated dynamical system is equipped with mass-action kinetics \cite{FeinbergLec79}, the concentration of $Z$ follows the equation
\begin{align}\label{eq:ACRode1}
\quad \frac{d}{dt}z(t)= \theta a(t) z(t) -\mu z(t).
\end{align}
At steady state one can set the right hand side equal to zero, and assuming $z\not=0$ one obtains that the steady state value for $A$ is $\frac{\mu}{\theta}$.
Letting $a(0)+z(0)=N$ be the initial input of the system, the only positive roots of the right-hand side of \eqref{eq:ACRode1} are $(a,z)=(\frac{\mu}{\theta},N-\frac{\mu}{\theta})$.
Hence this system is an ACR system and species $A$ is an ACR species.  See Figure 1b for a phase plane diagram illustrating this behaviour.  

In comparison to the ACR network in \eqref{eq:ACRcontroller}, we consider the simple reaction 
\begin{align}\label{eq:compare}
A\xrightleftharpoons[4]{2} B,
\end{align}  where the total mass of $A$ and $B$ is also invariant in time. In both systems the dynamics is confined to one of the black straight lines in Figure \ref{fig:example}a and \ref{fig:example}b.
The positive steady states of this system lie on the intersections between the nullclines (red) and the phase planes (black). Hence for system \eqref{eq:compare}, the steady state values of both $A$ and $B$ vary depending on the initial condition as shown in Figure \ref{fig:example}a. On the other hand, the ACR network system \eqref{eq:ACRcontroller} is such that all the positive steady states for $A$ are identical, as shown in Figure \ref{fig:example}b.

A similar type of control has been considered by Mustafa Khammash and others in Briat et. al. \cite{briat2016antithetic}. In that work, Khammash and colleagues propose an antithetic integral feedback circuit
\begin{align*}
Z_1+Z_2\xrightarrow{\eta} 0\xrightarrow{\mu}Z_1 \xrightarrow{k} Z_1+X_1, \quad  X_\ell \xrightarrow{\theta}X_\ell+Z_2,
\end{align*}
 that robustly stabilizes a species of interest, $X_\ell$, in the presence of intrinsic noise. In the controlled system the mean of the stochastic dynamics of a target species is stabilized at a pre-specified value with a low metabolic cost. A recent follow-up work \cite{aoki2019universal} experimentally implemented the antithetic control circuit in a growth-rate control system in \emph{E. coli}. We point out that this antithetic control circuit satisfies the topological conditions for ACR provided by Shinar and Feinberg in \cite{shinar2010structural}, and thus it is a specific example of an ACR system. In particular, the control species $Z_1$ and $Z_2$ are ACR species if the system admits a positive steady state. There are however a few important differences with our work. As we will show, the ACR controllers aim to control the mean, variance and even higher moments of a target species by controlling its distribution. The controllers in \cite{briat2016antithetic, aoki2019universal} are designed to robustly control the mean of the target species, but the controller might increase its noise. To account for the noise in the target species, Briat et. al. show in the follow-up work \cite{briat2018antithetic} that for a unimolecular model, an additional negative feedback loop can reduce the noise up to the original variance. Also, stochastic ACR controllers control the target distribution approximately assuming that sufficient copies of the control species present, while the control proposed by \cite{briat2016antithetic, aoki2019universal} provides exact control without approximation scheme.

   A particularly powerful property of an ACR system is that it can endow the ACR property to a given network. When an ACR system is added to a non-ACR model, one of the species in the combined system could become absolutely robust.  For example, suppose that species $A$ is present in a given deterministic network but that $Z$ is not, and that we add the two reactions \eqref{eq:ACRcontroller}. Then the dynamics of $Z$ will still satisfy $z’(t)= \theta a z -\mu z$. Moreover, at steady state it must still hold $a=\mu/\theta$ due to the same analysis as before.  Thus species $A$ is now absolutely robust in the new system.  

   If the ACR property of the ACR system is inherited to the target species in the controlled system, then we call the ACR system an \textit{ACR controller}. Throughout this paper we also call the newly introduced species in the ACR controller a \textit{control species}. In the following sections, we show that the steady state value of the target species is tunable with the parameters of the ACR controller. In the Supplementary Material, we further investigate the local stability of the steady state in the controlled system.

While the ACR controller has the ability to control a target species in deterministic systems, chemical species are often modeled as discrete entities. Stochastic models in biology have become increasingly relevant, as people have noticed that intrinsic noise significantly contributes to the dynamical behaviour \cite{Arkin1998, Becskei2005, Elowitz2002, Paulsson2011,Maamar2007, ULPOSP2016}. The effects of noise are especially large if the abundance of a species in the system is low. Many important biochemical models consist of species with low copy numbers inside each individual cell \cite{Paulsson2004}. More details on the modeling of stochastic networks are included in the Supplementary Material. 
In stochastic models we have additional control goals than for deterministic models, as it is important to not only control the mean expression level but also its variance (i.e. noise) and ideally the full probability distribution of the target species.
If a thermostat makes the temperature of a room oscillate between $30^{\circ}$ F and $110^{\circ}$ F, one might argue that the room is controlled with mean temperature $70^{\circ}$ F, although its occupants might disagree. 

In order to control stochastic systems, we rely on the mathematical theory of deficiency in chemical reaction networks. The \emph{deficiency} of a reaction network is a non-negative integer that is determined by the topology of the network regardless of parameter values. Networks with deficiency zero and a weak reversibility property have well characterized long term dynamics, under both deterministic and stochastic conditions. In the deterministic case, such systems admit a unique local asymptotically stable steady state for given total amounts of the species \cite{Horn72, Feinberg72,craciun2015toric}. For a stochastic system, under the same conditions, each of the species has Poisson distribution centered around its deterministic steady state \cite{AndProdForm} (see the Supplementary Material for additional details).
 These strong properties inspire us to propose a new deficiency based control scheme for stochastic reaction networks, based on recent work expanding ACR to the stochastic case \cite{Enciso2016, anderson2017finite}. 

One property observed in some stochastic chemical reaction networks is a so-called \emph{extinction event}. Such an event takes place when some of the species disappear and can never return to the system. A stochastically modeled ACR controller can go extinct if a control species, such as $Z$ in the basic ACR controller, is entirely removed from the system.  This phenomenon is commonly present in ACR networks \cite{AEJ2014}.  One way to minimize this effect is to run the controllers with sufficiently high control species abundance, so that a potential breakdown of the ACR system is rare.

This high abundance setting is indeed a suitable assumption for the study of stochastic systems \cite{Kurtz72}. In our stochastic systems, each species can be categorized as either high abundance or low abundance, compared with the total protein abundance $N$. We use $N$ as a scaling parameter, and we carry out a multiscaling procedure to reach a reduced stochastic reaction network. By assuming that the reduced network has zero deficiency and is weakly reversible, we conclude that the target species has approximately Poisson distribution both in the reduced and the original networks. In special cases, the reduced model can be treated as a hybrid between deterministic and stochastic networks \cite{ anderson2017finite,ball2006asymptotic,  KangKurtz2013}.

\begin{figure}[h!t]
\centering
    \includegraphics[]{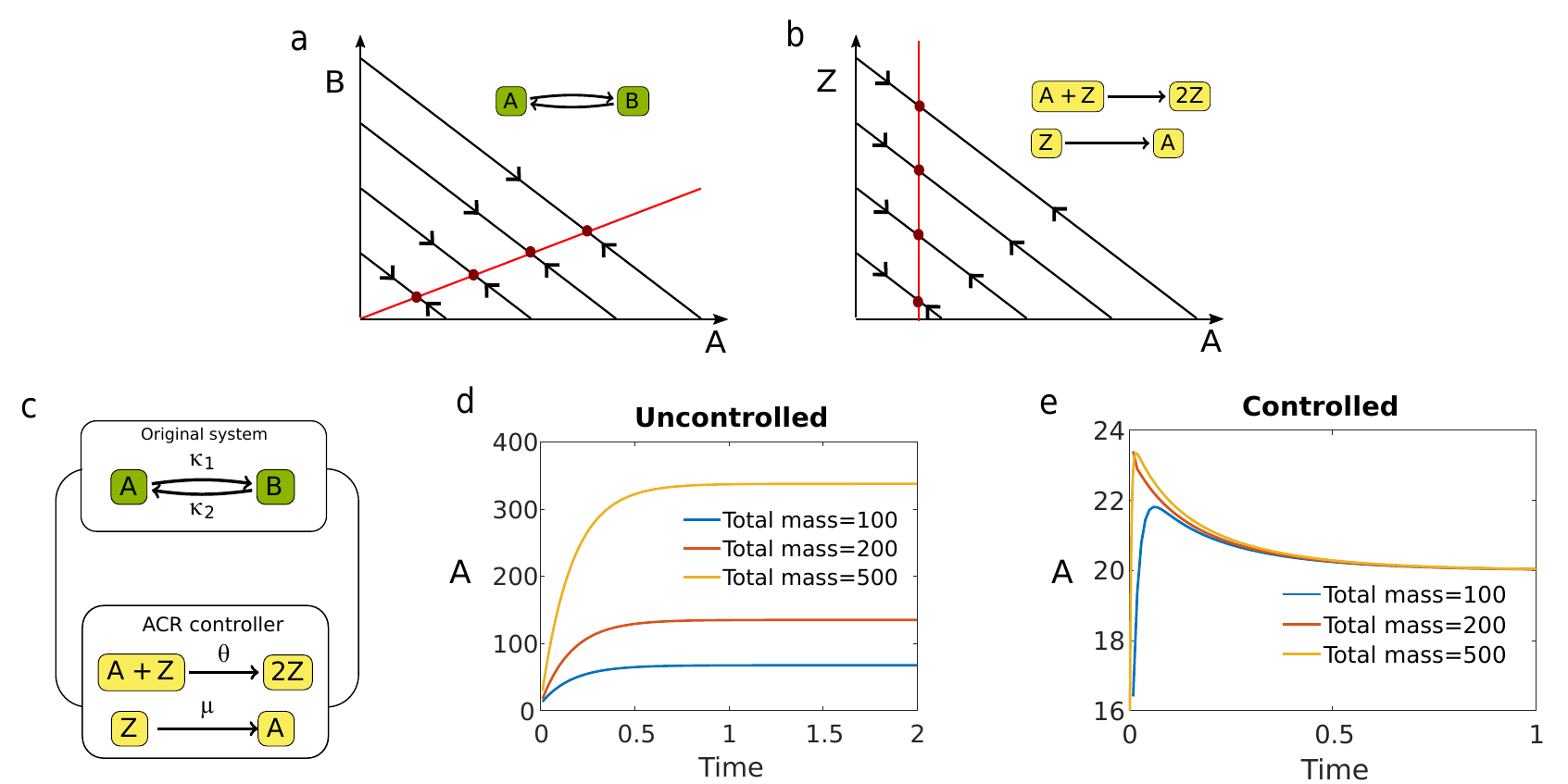}
    \caption{\textbf{a.} and \textbf{b.} Dynamics of the networks $A\rightleftharpoons B$ and the basic ACR controller, respectively. The intersection between each black line and the red line is a steady state for a given total mass. \textbf{c.} The original model in 1a is controlled using the basic ACR controller. \textbf{d.} Time evolution of $A$ in the original system, $k_1=2,k_2=4$.  \textbf{e.} Time evolution of $A$ in the controlled system via the ACR controller, using the above parameters and $\theta=1$, $\mu=20$.} \label{fig:example}
\end{figure}

We provide multiple examples of control of given biochemical networks using different ACR controllers.  For example, we use an existing deterministic model of ERK signal transduction from Rubinstein et al \cite{rubinstein2016long}, and we stabilize the dose response of this system using the basic ACR controller.  We also study a stochastic receptor-ligand model in which we target the concentration of free receptor, and we show that the concentration of a downstream regulatory protein is also controlled as a result.  Finally, we study a stochastic dimer-catalyzer model together with an expanded ACR controller as an application of the hybrid approach.  Simulations using the Gillespie algorithm are provided throughout to illustrate the control implemented by our approach. We also emphasize that the controlled networks admit a robust perfect adaptation property for both deterministic and stochastic examples.

\section{Results}

\subsection{Deterministic Control Using ACR Networks}
\label{subsec:ACR control}

We begin by using an ACR controller for a deterministic system. Consider again the simple translation model between species $A$ and $B$,
 \begin{align}\label{eq:simple translation}
 A \xrightleftharpoons[k_2]{k_1} B.
 \end{align} 
Letting $a(t)$ and $b(t)$ denote the concentration of species $A$ and $B$, respectively, the associated deterministic system with mass-action kinetics is
\begin{align}\label{eq:simple example}
\frac{d}{dt}a(t)=-k_1a(t)+k_2 b(t), \quad \frac{d}{dt}b(t)=k_1a(t)-k_2 b(t).
\end{align} 
We notice that $\frac{d}{dt}a(t)+\frac{d}{dt}b(t)=0$ which implies that the total mass $a(t)+b(t)$ is conserved. 
When $a(0)+b(0)=N$, the steady state of the system is $(a^*,b^*)=(\frac{\kappa_2}{\kappa_1+\kappa_2}N,\frac{\kappa_1}{\kappa_1+\kappa_2}N)$ by using the conservation $a(t)+b(t)=N$. Hence the positive steady state concentration of $A$ in the original system \eqref{eq:simple translation} varies along with the initial input $N$. To get the desired steady state value for $A$, therefore, fine-tuning of the initial condition $N$ is necessary.

Adding the basic ACR controller \eqref{eq:ACRcontroller} to the original system, we have a new system
\begin{align}\label{eq:controlled system example}
Z+A\xrightarrow{\theta} 2Z, \quad Z\xrightarrow{\mu}A \xrightleftharpoons[k_2]{k_1} B.
\end{align}

As described in the introduction, using the equation for $Z$ we can deduce that for any positive steady state it must hold $a^*=\frac{\mu}{\theta}$. See Theorem S3.2 in the Supplementary Material for a generalization of this statement to other networks as well as systems with reaction kinetics different from mass action.  

As an application, we use the basic ACR network to control a system of an extracellular signal regulated kinase (ERK) activation shown in Rubinstein et al \cite{rubinstein2016long}. ERK is a widely studied protein in signal transduction, and it is activated through phosphorylation at two different sites. The steps of the dual phosphorylation are regulated by other protein kinases \cite{rubinstein2016long}. We denote ERK by $S$, and consider the four phosphorylation forms $S_{00},S_{01},S_{10}$ and $S_{11}$ depending on the phosphorylated sites.  Nonsequential ERK phosphorylation is mediated by mitogen-activated protein kinase MEK, denoted here by $E$.  The variable $F$ denotes a nonspecific phosphatase that mediates ERK dephosphorylation. The steps of phosphorylation and dephosphorylation are described with the reaction network model in Figure \ref{fig:ERK}a.

In the ERK system in Figure \ref{fig:ERK}a, there are three conservation relations.  For instance, $E_{tot}=E+[ES_{00}]+[ES_{01}]+[ES_{10}]$ represents the total concentration of kinase, and similarly for total substrate $S_{tot}$ and total phosphatase $F_{tot}$. It has been shown that the mass action deterministic model associated with the ERK model in Figure \ref{fig:ERK} a has different long-term dynamical behaviour depending on the system parameters \cite{conradi2015global, obatake2019oscillations, rubinstein2016long}. These dynamical behaviours include unique stable stationarity, sustained oscillations, and bistability. We use the parameters in Rubinstein et al., which are such that the ERK system converges to a unique, stable, and positive steady state \cite{rubinstein2016long}. 

One of the most important features of this system is its so-called dose response, which describes active ERK $S_{11}$ as a function of the kinase intput $E$. However this dose response depends on the total amount of phosphatase $F_{tot}$. We introduce a control using the basic ACR controller in order to fix $S_{11}$ for every given value of total kinase $E_{tot}$, and therefore to stabilize the dose response.  As the plot on the right-hand side of Figure \ref{fig:ERK}b shows, the steady state concentration of protein $S_{11}$ is sigmoidal as a function of $E_{\text{tot}}$ for fixed $F_{\text{tot}}$. The goal of control with the basic ACR system in Figure \ref{fig:ERK}a is to equalize the positive steady state of the phophatase $F$ for any $F_{\text{tot}}$, and eventually to obtain the same sigmoidal curve for the steady state concentration of $S_{11}$, see Figure 2c.  

Aside from the mathematical model, it is important to think how the basic ACR system in Figure \ref{fig:ERK}a could be implemented experimentally. There are several possible approaches which might depend on the individual system.  In this case, suppose that the phosphatase $F$ is bifunctional, acting as a phosphatase in its standard form and turning into a kinase $Z$ when it is itself phosphorylated. Suppose that kinase $Z$ mediates the phosphorylation of protein $F$ as depicted with the reaction $F+Z \to 2Z$.  Finally, another phosphatase, which is not explicitly modeled in this system, eventually dephosphorylates $Z$ into $F$ as described with the reaction $Z\to F$.  This set of assumptions would suffice to implement the control network. Notice that bifunctional enzymes can be found in the literature, for instance EnvZ in \emph{E. coli} osmolarity regulation \cite{batchelor2003robustness}. Notice also that the self-mediated phosphorylation can be found in the epidermal growth factor receptor (EGFR) \cite{normanno2006epidermal, needham2016egfr}. While the possibility of the practical implementation of the basic ACR controller remains open, we focus on theoretic aspects of the controller in this manuscript.

Assuming that the dynamics associated with the ERK model and the basic ACR system in Figure \ref{fig:ERK}a follows mass-action kinetics, we use $\mu=2$ and $\theta=1$. This implies that for any input $E_{\text{tot}}, F_{\text{tot}}$ and $S_{\text{tot}}$, the steady state of $F$ is $2$. The convergence of $F$ to $2$ is also theoretically proven, see Section S5.1 in the Supplementary material. Thus in the controlled system, the concentration of $F$ converges to $2$, unlike the uncontrolled original ERK model which has different steady state concentrations of $F$ for different values of $F_{\text{tot}}$, as described in Figure \ref{fig:ERK}b (left) and \ref{fig:ERK}c (left). As a result, $S_{11}$ in the controlled system has identical dose response regardless of the value of $F_{tot}$ (right plot in Figure \ref{fig:ERK}c). On the other hand, the $S_{11}$ dose responses are different in the original system (the right plot in Figure \ref{fig:ERK}c).
 
\begin{figure}[H]
\centering
	\includegraphics[]{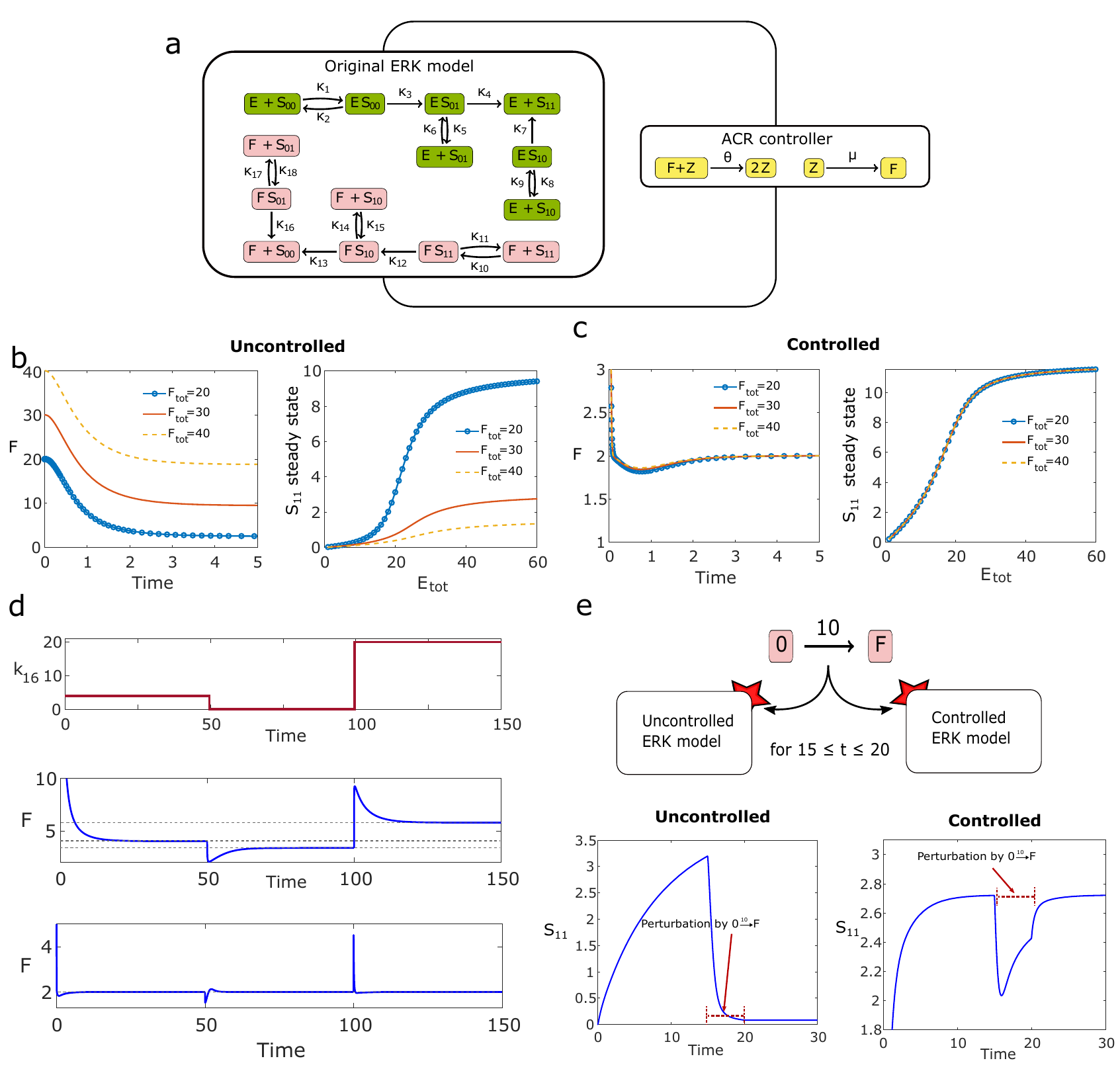}
	\caption{\footnotesize{ \textbf{a.}  Reaction network of an ERK system originally proposed in \cite{rubinstein2016long}, together with the basic ACR controller.  Parameters used are $k_1=3,k_2=2,k_3=1,k_4=2,k_5=1,k_6=3,k_7=2,k_8=5,k_9=3,k_{10}=2,k_{11}=2,k_{12}=3,k_{13}=1,k_{14}=3,k_{15}=1,k_{16}=1,k_{17}=4, k_{18}=4$ for the original ERK system, and $\theta=1$, $\mu=2$ for the ACR controller. \textbf{b.} (Left) Time evolution of $F$ in the original system without the ACR controller. (Right) Dose response of active ERK $S_{11}$ as a function of total kinase $E_{tot}$ for the original system. Initial conditions are $S_{00}=50$, $E=E_{tot}$, $F=F_{tot}$, and zero for all other species. 
	 \textbf{c} (Left) Time evolution of $F$ in the controlled system. (Right) Dose response of $S_{11}$ as a function of $E_{tot}$ for the controlled system. Initial conditions are $S_{00}=50$, $Z=50$, $E=E_{tot}$, $F=F_{tot}$, and zero for all other species.
\textbf{d.} Robust perfect adaptation of $F$ in the controlled system. We set the initial conditions as $S_{00}=30, E_{\text{tot}}=30, F_{\text{tot}}=30, z(0)=30$ and zero for all other species.
\textbf{e.} The dynamics of $S_{11}$ in the original (left) or controlled
(right) system which is transiently perturbed by reaction $0\xrightarrow{10}$ $F$ for $t\in [15,20]$.}} \label{fig:ERK}
\end{figure}

\subsection{Robust Perfect Adaptation}\label{subsec:robust_deter}

One of the main purposes of the control with an ACR system is to create 
\emph{robust perfect adaptation} for the target species. A system possesses perfect adaptation if the output of the system returns to the pre-stimulus level after the value of the input parameter is changed to a different constant level. Furthermore, if the system achieves perfect adaptation independently of the system parameters, it is said to have robust perfect adaptation (RPA) \cite{ferrell2016perfect, ma2009defining,  xiao2018robust, yi2000robust}.

 Here we show that species $F$ in the controlled ERK model in Figure \ref{fig:ERK} admits RPA. We persistently disturb the parameters by changing them in time such as $k_{16}$ as shown in Figure \ref{fig:ERK}d (top). As expected, the concentration of phosphatase $F$ converges to different values for each perturbation as depicted in Figure \ref{fig:ERK}d (middle). However as Figure \ref{fig:ERK}d (bottom) shows, the controlled ERK system has robustness to the perturbations for $F$ as its concentration converges to the set-point $\frac{\mu}{\theta}=2$ regardless of parameter values. This RPA for $F$ basically arises because the steady state of $F$ is completely determined by the two reactions in the ACR controller, and it is independent on the parameters $k_i$ of the original ERK system. More details are in Section S3 of the supplementary material.

In addition to RPA, we also show that the controlled system is robust to a transient change of the system structure. We perturb the system by turning on an additional reaction $0\xrightarrow{10} F$ for time $[15,20]$. The $S_{11}$ concentration in the uncontrolled system initially converges, but it immediately responds to the transient in-flow as $F$ is produced during that time interval (Figure \ref{fig:ERK}e (left)). The concentration does not return to the previous steady state value after the transient perturbation is turned off. In the controlled system, the phosphorlyated protein $S_{11}$ also responds when the transient in-flow is switched on at $t=15$ as show in Figure \ref{fig:ERK}e (right). However, it is quickly driven back to the steady state after the perturbation is switched off at $t=20$.  
This robustness to the transient structural disturbance stems from the fact that the controlled system admits a single positive steady state concentration for the target species $F$ regardless of the input states. Hence when the transient inflow is turned off, $F$ in the controlled system converges to the set-point wherever the current state of the system has been driven by the perturbation.

\subsection{Control of Additional Species}\label{sec:additional species}

Recall that the basic ACR system \eqref{eq:ACRcontroller} consists of the control species $Z$ and the target species $A$. The fact that $Z$ only directly controls $A$ may impose limitations in some situations. We show in the following example that an ACR controller with reactions involving other network species can provide better performance.

For example, consider the following reaction network where no conservation relations exist:
\begin{align}\label{eq:unbound system}
\begin{split}
&\hspace{0.6cm} A	\\[-1ex]
\text{\scriptsize{$\kappa_2$}} \hspace{-0.2cm}&\nearrow \\[-0.6ex]
A+B \xrightarrow{\kappa_1} \emptyset & \\[-1ex]
\text{\scriptsize{$\kappa_3$}} \hspace{-0.2cm}&\searrow  \\[-1ex]
&\hspace{0.6cm} B
\end{split}
\end{align}

In this system, two proteins $A$ and $B$ are constantly produced but also degrade each other. Using the parameters $\kappa_1=1, \kappa_2=3$, $\kappa_3=5$, the concentration of $A$ decays toward zero as shown in Figure \ref{fig:AB}d. Despite the addition of the basic ACR system, the concentration of $A$ still decays to zero as shown in Figure \eqref{fig:AB}e.  See Section S5.2 in the Supplementary Material for additional details about this system, including the existence of positive steady states.
 
We design the expanded controller shown in Figure 3c to include both $A$ and $B$ in the reactions.  It can be verified that the mass-action system associated with this controller is ACR, with ACR species $A$. This is because the additional reactions $Z+B \rightleftharpoons Z$ do not change $Z$, so that the equation for $Z$ is the same as in the base ACR model. Such reactions that have no contribution to the control species are also used for the antithetic integral controller in \cite{briat2016antithetic, aoki2019universal}. Using $\theta=1,\mu=5,\alpha_1=2$, $\alpha_2=1$, one can see that this controller steers the positive steady state concentration of $A$ to $5$ for different initial conditions (Figure \ref{fig:AB}f). Here the $\alpha_1$ and $\alpha_2$ need to be chosen in a certain range. See Section S4 in the Supplementary Material for  more details about control of general 2-dimensional systems.

\begin{figure}[H]
\centering
	\includegraphics{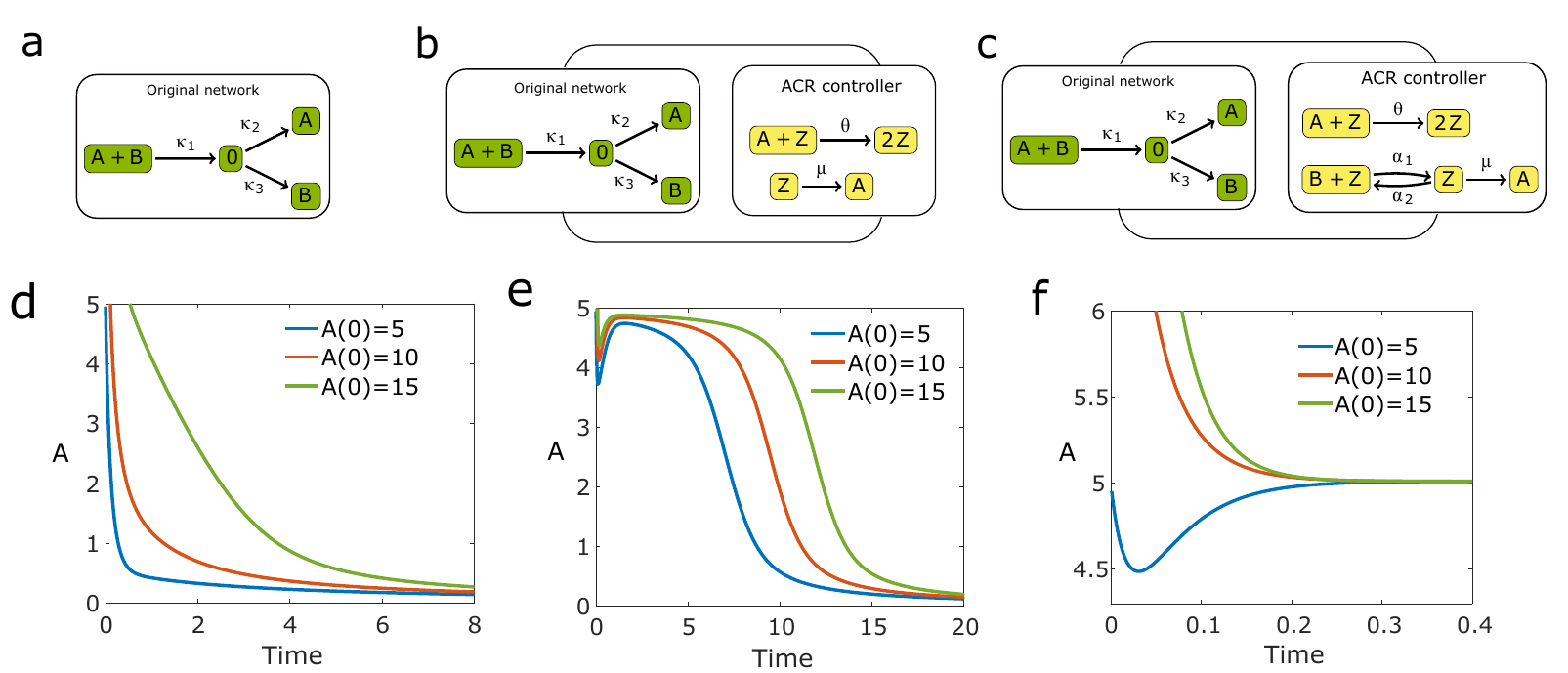}
	\caption{ \footnotesize{\textbf{a.} Original network, using parameters $\kappa_1=1, \kappa_2=3$ and $\kappa_3=5$. \textbf{b}. The basic ACR system is added to the original system, with $\mu=5$, $\theta=1$. \textbf{c.} Expanded ACR controller, with $\alpha_1=2$ and $\alpha_1=1$. \textbf{d.} The concentration of $A$ converges to zero in the original system. \textbf{e.} The basic ACR system in \ref{fig:AB}b fails to control $A$, as $A$ still converges to zero for different initial values. \textbf{f.} The concentration of $A$ is driven to the set-point $5$ with each initial condition, for the controlled system in c.}} \label{fig:AB}
\end{figure}

\subsection{Stochastic Control Using an ACR Module}
\label{subsec:stochastic control}

When a system contains species with low copy numbers, the intrinsic noise considerably affects the system dynamics. Therefore we model the system stochastically using a Markov process. This continuously evolving Markov process defined on a multi-dimensional integer grid has state-dependent transition rates (for more detail see Section~S1.1 in the Supplementary material). In the context of stochastic control, recent work by Mustafa Khammash and others has proposed controlling a target species by adding four reactions.  While that framework allows to control the mean of the target species, there could be significant variability in its noise. Our ACR approach makes use of topological properties of the original network to approximate the full distribution of the target species.  

In stochastic networks, if one of the species reaches zero copies, then a subset of the reactions in the system would be turned off, potentially preventing the species from ever being produced.  Such an extinction event can take place for $Z$ in the basic ACR controller as well as many other ACR systems \cite{AEJ2014}. In order to avoid this situation, we design the basic ACR system with sufficiently high copies of the control species. More generally, we assume that all species are classified into two types: highly abundant species such as control species $Z$ which are of order $N$ for a scaling parameter $N$, and low abundance species of constant order. We also scale the parameters of the controlled system to make all the reaction propensities of constant order. Under this same scaling, Enciso \cite{Enciso2016} used the technique of species `freezing' for an ACR system to generate a reduced network of low abundant species. It was further shown that if the reduced network has zero deficiency and is weakly reversible, then an ACR species of low order tends to follow a Poisson distribution centered at its ACR value, as time $t$ and the scaling parameter $N$ go to infinity. 

The work in \cite{Enciso2016} approximated the distribution of the target species with the help of a reduced stochastic model, which is the limit of the original stochastic network using a multiscaling procedure.   Similar types of approaches have been studied using different system scaling, network topological conditions or state space truncations \cite{anderson2017finite, enciso2019embracing, Gupta_Khammash2017,  StoQSSA2019,  kim2017reduction, sontag2017reduction,  MunskyKhammash2008}. The multiscaling assumption in \cite{Enciso2016} is somewhat special in that all reaction propensities have  constant order of magnitude up to finite time. 

Given a stochastic chemical reaction network, we now add an ACR controller and use the scaling procedure described above in order to study the resulting controlled system. To exemplify this we consider a model describing the dynamics of a receptor binding to a ligand and generating a downstream response (Figure \ref{fig:RL schematic}a). Many important biology models involve receptor-ligand interactions such as signal transduction, physiological regulation, and gene transcription.  In this case a ligand $L$ binds to an inactive receptor $R_0$ on the cell membrane, converting it into an active receptor $R$. Two active receptors are dimerized, forming the species $D$ which is phosphorylated sequentially in three different locations. The triphosphorylated dimer $D_3$ transmits the signal inside the cell by activating another protein $P$ as shown in Figure \ref{fig:RL schematic}d. We control the inactive receptor $R_0$ using the basic ACR system, in order to control the desired amount of active protein $P^*$.

Once again, the practical implementation of such a system must depend on the specific receptor.  We suggest a possible implementation as follows: suppose that a second ligand, called an antagonist, binds to the receptor forming a molecule $Z$, which prevents the binding of the original ligand (see Figure \ref{fig:RL schematic}d).  Suppose the complex $Z$ facilitates the recruitment of another antagonist to produce another copy of $Z$, leading to the reaction $Z+R_0\to 2Z$. The reaction $Z\to R_0$ simply represents the natural unbinding of the antagonist from $R_0$. Another option could be to think of $Z$ as a misfolded form of $R_0$, and of the reaction $Z+R_0\to 2Z$ as a prion-like effect where a misfolded receptor makes it more likely that a second receptor will misfold. In any of these cases, the introduction of an new molecule into the system (the antagonist or the misfolded protein) leads to two additional reactions that control the network.     

\begin{figure}[H]
\centering
\includegraphics{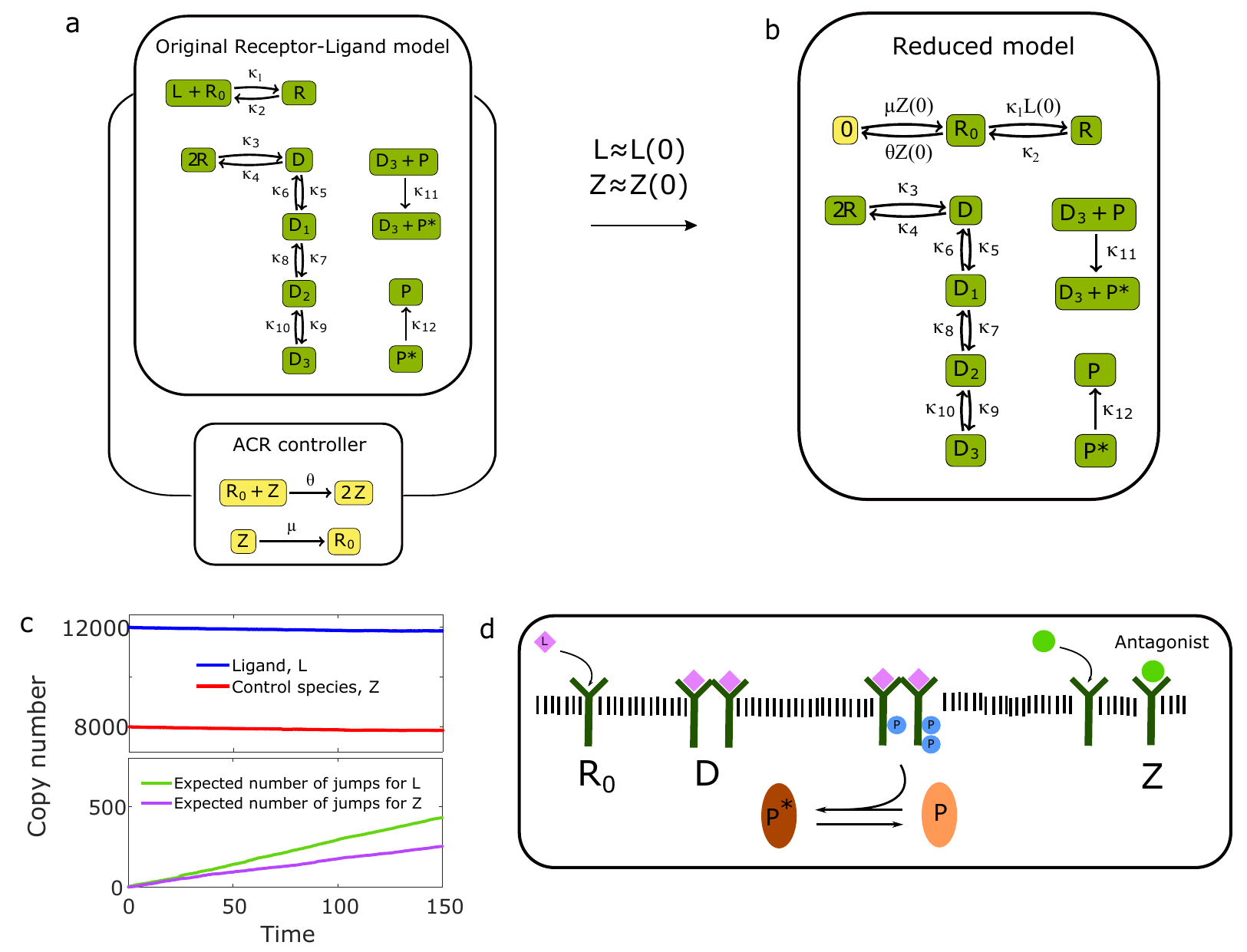}
\caption{\footnotesize{\textbf{a.} Reaction network for the receptor-ligand  pathway (green) and the ACR controller (yellow). \textbf{b.} Reduced model obtained by freezing $L$ and $Z$ at their initial values, respectively. \textbf{c.} Stochastic time evolution of the copy numbers of $L$ and $Z$, highlighting the small net change of $L$ and $Z$ in the system by time $t=150$. \textbf{d.} A schematic picture for the receptor-ligand model and the ACR controller}}\label{fig:RL schematic}
\end{figure}

We let the system start with initial counts $L(0)=1500, Z(0)=1000, R_0(0)\le 50$ and the initial copy numbers of all the other species equal to zero. Hence species $L$ and $Z$ are the high abundance species of order $N=1000$ and the other species are of low abundance. 

As mentioned above, the main idea of the control for this system is to approximate the distribution of $R_0$ by the reduced network in Figure \ref{fig:RL schematic}b, which we now explain. Parameters are chosen as $\kappa_1=0.82\times 10^{-3}, \kappa_2=1.37, \kappa_3=1.41, \kappa_4=1.79,  \kappa_5=1.02, \kappa_6=1.36, \kappa_7=1.97, \kappa_8=1.11, \kappa_{9}=1.55, \kappa_{10}=1.01,\kappa_{11}=1.34, \kappa_{12}=0.5, \theta=10^{-3}$ and $\mu=5\times 10^{-3}$. In order to arrive to this parameter set, parameters $\kappa_2$ through $\kappa_{12}$  were randomly chosen in the range $[1,2]$.
Parameters $\kappa_1,\theta$ and $\mu$ are associated with reactions involving high abundance species $L$ and $Z$, and they were chosen of order $\frac{1}{N}$ so that the reactions $L+R_0\to R, Z+R_0\to 2Z$ and $Z\to R_0$ have constant order propensities under mass-action kinetics. Details of the mass-action propensity computations are provided in Section S7.1 in the Supplementary Material. Because of the low propensities of the reactions relative to $N$, the expected change of species $L$ and $Z$ by $t=150$ are much smaller than the copy numbers of $L$ and $Z$ as Figure \ref{fig:RL schematic}c shows. By neglecting the relatively small number of fluctuations for $L$ and $Z$ shown in Figure \ref{fig:RL schematic}c, we can freeze them at their initial counts and obtain a reduced system in Figure \ref{fig:RL schematic}b.

\begin{figure}[H]
\centering
\includegraphics{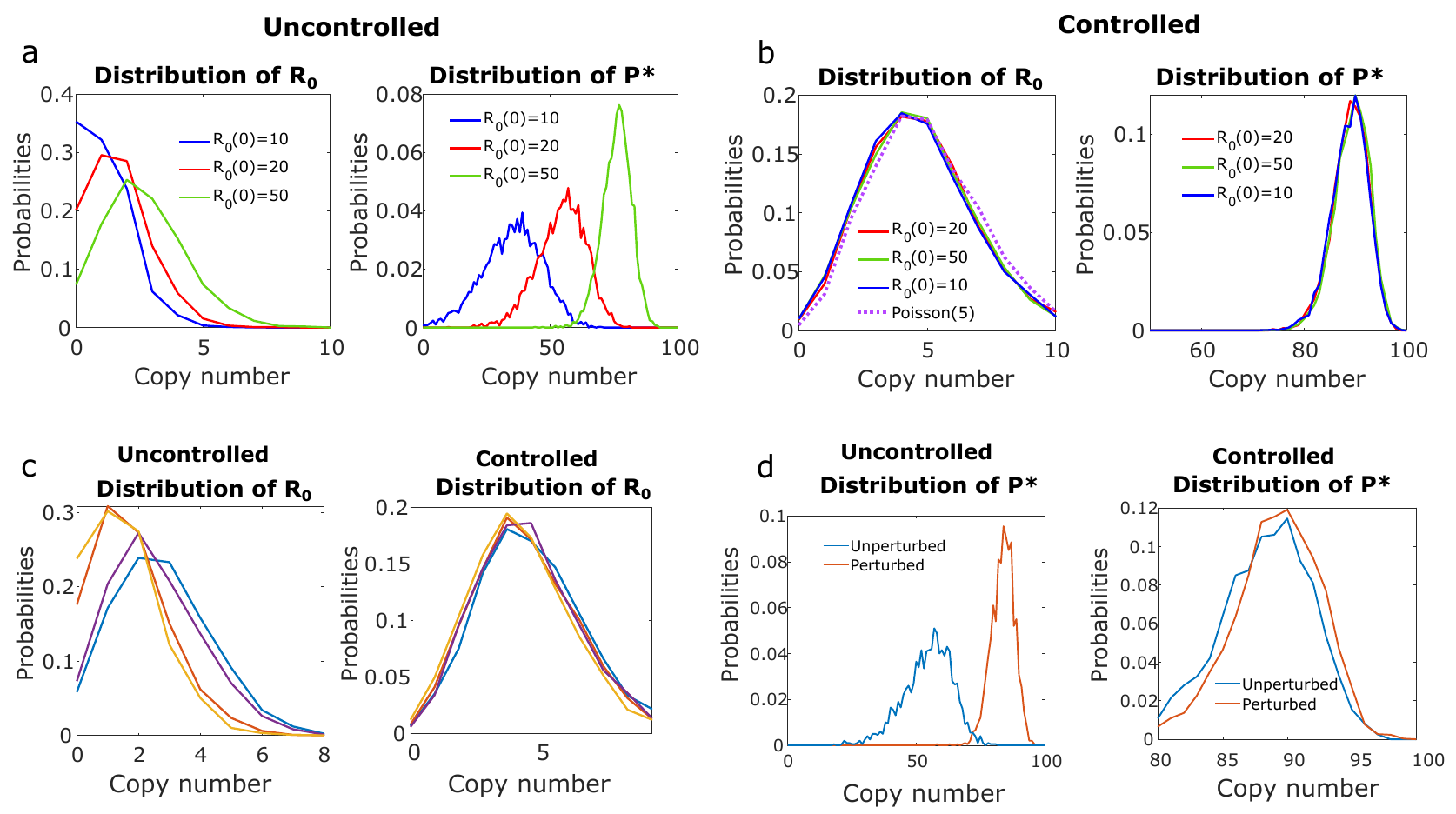}
\caption{ \footnotesize{Gillespie simulations \cite{Gillespie77} for the distribution of $R_0$ and activated protein $P^*$. We use the initial values $L(0)=1500, P(0)=100, R(0)\le 50$, and the remaining species have zero initial values. For the ACR controller, we set $Z(0)=1000$.
\textbf{a.}  For the uncontrolled system, distribution at time $t=150$ of inactive receptor $R_0$ (left) and active protein $P^*$ (right). 
\textbf{b.} (Left) For the controlled system, distribution at time $t=150$ of $R_0$ (left) and $P^*$ (right).
\textbf{c} and \textbf{d.} Robustness of the system to randomized parameter perturbations and a transient reaction perturbation, setting $R_0(0)=20$. \textbf{c.} Four distributions of $R_0$ obtained by simulations of the uncontrolled
(left) and controlled (right) receptor-ligand model with randomly
perturbed system parameters. \textbf{d.} Unperturbed and perturbed distributions of $P^*$ by transiently
switched-on reaction $0\xrightarrow{2} R_0$ for $t\in [50,80]$ in the uncontrolled (left) and controlled
(right) receptor-ligand model.}}\label{fig:stochasticLR} 
\end{figure}
 
 Recalling that $R_0$ in the receptor-ligand network in Figure \ref{fig:RL schematic}a and its associated reduced network \ref{fig:RL schematic}b behave similarly over time, we carry out an analysis of the reduced network. We initially  ignore the reactions $D_3+P \to D_3+P^*$ and $P^*\to P$ because they only affect the proteins $P,P^*$ without an effect on other species. The resulting network is reversible, and it has zero deficiency since
 \begin{align*}
 n-\ell-s=8-2-6=0,
 \end{align*}
 where $n$ is the number of complexes, $\ell$ is the number of connected components, and $s$ is the rank of the stoichiometry matrix. The distribution of $R_0$ in the reduced network converges to a Poisson distribution in the long run \cite{AndProdForm}. Depending on the number of states of the stochastic system, the long-term distribution could be a truncated Poisson distribution. However, the reaction in the ACR controller are always reduced to the inflow and the outflow of the target species in the reduced system. Hence the long-term distribution is always approximated a Poisson distribution. See Section S6.3 in the Supplementary Material for a more rigorous statement.

Furthermore, the rate of the Poisson distribution associated with $R_0$ is determined by the steady state value of the corresponding deterministic system \cite{AndProdForm}. Since the rate is equal to the mean of the Poisson distribution, the mean of $R_0$ is close to $\frac{\mu}{\theta}$ in the long run.  Thus species $R_0$ in the controlled system shown in Figure \ref{fig:RL schematic}a at a sufficiently large finite time $t$ is well approximated by the Poisson distribution centered at $\frac{\mu}{\theta}$. This is shown in Figure \ref{fig:stochasticLR}b (left) at $t=150$, where for any input $R_0$ the distribution seems almost Poisson($\frac{\mu}{\theta}$).  Consequently the protein $P$ distribution is also robustly stabilized as shown in Figure \ref{fig:stochasticLR} b (right). On the other hand both mean and variance of $R_0$ in the original system vary with respect to different inputs (Figure \ref{fig:stochasticLR}a, left), and this causes the distribution of $P^*$ to change accordingly (Figure \ref{fig:stochasticLR}a, right). 

In an additional analysis, we study the convergence speed of the distribution of the reduced system towards a stationary distribution in Section S7.1 of the Supplementary Material. The underlying mathematical framework, with an emphasis on the accuracy of the approximation between the controlled network and the reduced system, is further described in our follow-up paper \cite{EK2019}.

For the receptor-ligand system, the basic ACR module also robustly controls the target species to perturbations. We perturb the parameters $\kappa_i$ in Figure \ref{fig:stochasticLR} using the equation $\kappa'_i=\kappa_i+r_i$, where the $r_i$ are sampled from a uniform distribution on the interval $[0,3]$. As shown in Figure \ref{fig:stochasticLR}c (left), the uncontrolled system generates distinct distributions of $R_0$ at $t=150$ for randomly perturbed parameters in each simulation. On the other hand, the distributions of $R_0$ at $t=150$ generated by the controlled system with the same parameters closely approximate the Poisson distribution with mean $\frac{\mu}{\theta}=5$, as shown in Figure \ref{fig:stochasticLR}c (right).

Plots in Figure \ref{fig:stochasticLR}d show how $P^*$ robustly behaves with a transient perturbation in the controlled system. We perturb the system with a reaction $0\xrightarrow{2} R_0$ only for time $t\in [50,80]$. Because of this additional input, the distribution of $P^*$ at $t=150$ is shifted to the right for the uncontrolled system (Figure \ref{fig:stochasticLR}d, left). However for the controlled system, Figure \ref{fig:stochasticLR}d  (right) shows that its distribution is robust to the transient perturbation.

\subsection{Stochastic Control Using a Hybrid Approximation}\label{subsec:faststo}

Recall that in the receptor-ligand system in Section \ref{subsec:stochastic control}, the fluctuation of species $L$ and $Z$ in the concentrations are negligible since the reaction propensities are small compared with their concentration. However, many classical studies of stochastic systems eliminate this assumption of small reaction propensities, see for instance the classical work by Kurtz \cite{Kurtz72}. Reaction propensities could also have different orders of magnitude with respect to $N$. In such cases, the stochastic system is modeled under a multiscaling regime, and its behaviour can be studied using a hybrid deterministic-stochastic system \cite{anderson2017finite, ball2006asymptotic, KangKurtz2013, preziosi2006hybrid, ge2015stochastic, lin2018efficient, ge2018relatively}. In a hybrid system, the counts of some species change stochastically while  the concentrations of the other species change continuously. We modify the basic controller in order to control such a hybrid system. In this section, using the finite time stationary distribution approximation in \cite{anderson2017finite}, we show that an expanded basic ACR system can be used to control a stochastic system under more general scaling.

As an example, we provide a dimer-catalyzer model in Figure \ref{fig:faststochastic}a. In this system the initial copy number of species $X^*, X_1, C, C_p$ and $C_{pp}$ are all of order $N=1000$. Hence using mass action kinetics all the reactions have order $N$ propensities. Using the framework established by Anderson et al \cite{anderson2017finite}, we approximate the original model with the hybrid system in Figure \ref{fig:faststochastic}d. The stochastic part of the hybrid system has zero deficiency and is weakly reversible so that the distribution of the target species $X$ is Poisson at a finite time $t=5$ \cite{anderson2017finite, AndProdForm}. The stochastic and deterministic parts are coupled, as the mean $m(t)$ of $X$ at finite time $t$ is determined by the dynamics of the deterministic system as depicted in Figure \ref{fig:faststochastic}e. A flux balance analysis implies that $m(t)=\dfrac{k_1x^\ast(t)+\mu z(t)}{k_2c_{pp}(t)+\theta z(t)}$.  In Section S7.2 of the Supplementary Material, we show that $m(t)$ converges to the desired value $\dfrac{\mu}{\theta}$, as $t\to \infty$, when the initial concentration of $Z$ is sufficient. The mean $m(t)$ actually converges to $\dfrac{\mu}{\theta}$ quickly as shown in Figure \ref{fig:faststochastic}e. Therefore unlike the distribution of $X$ in the original system as shown in Figure \ref{fig:faststochastic}b, the distribution of $X$ in the controlled system is approximately Poisson centered at $\dfrac{\mu}{\theta}=5$ at time $t=5$ for randomly sampled parameters $\kappa_i$ (\ref{fig:faststochastic}c).

\begin{figure}[h!bt]
\centering
\includegraphics[]{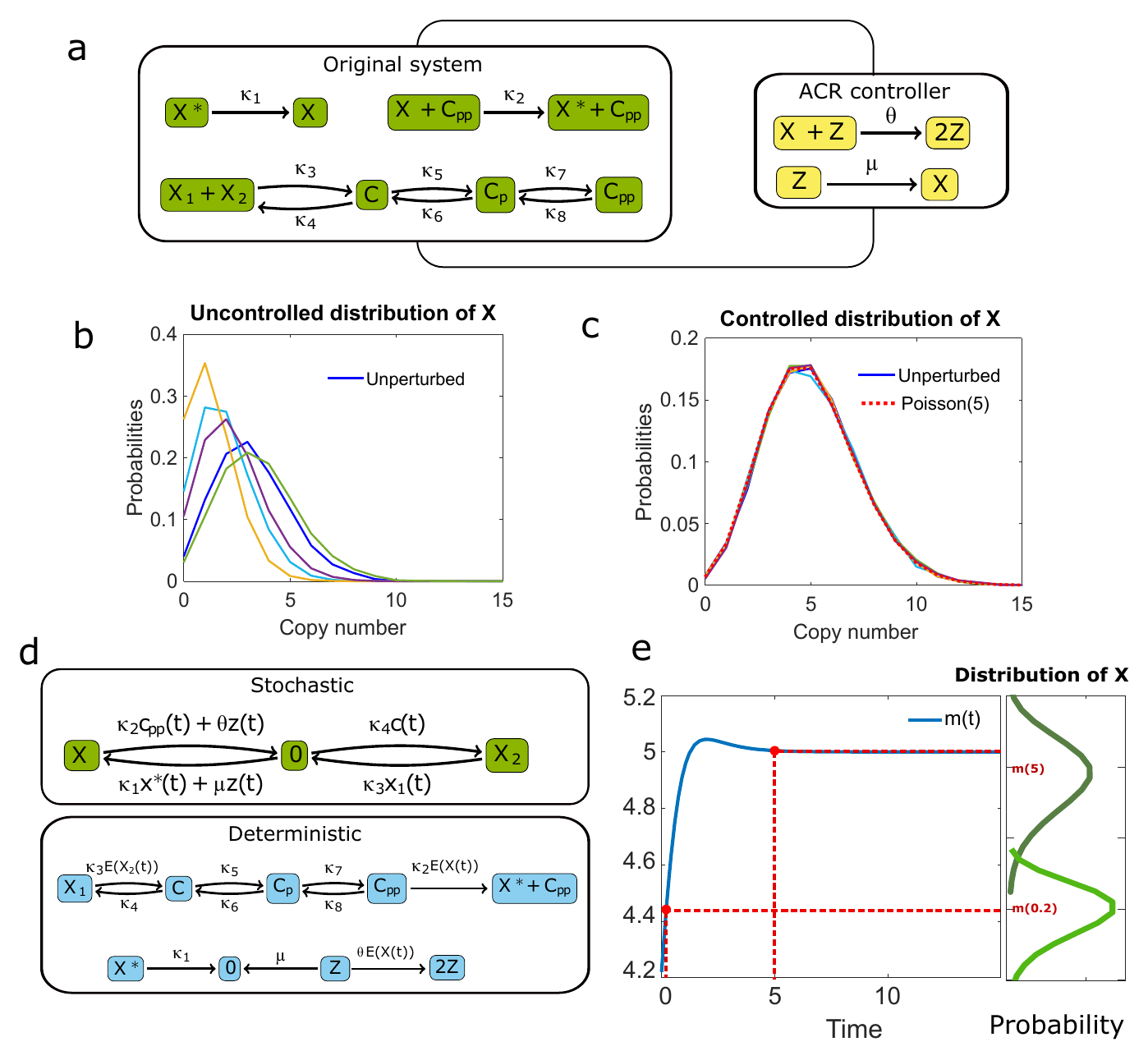}
 \caption{\footnotesize{Dimer-catalyzer model with high reaction rates of order $N$ \textbf{a.} The original model and the ACR controller. Parameters are $\kappa_1=1.38, \kappa_2=1.58, \kappa_3=1.19, \kappa_4=1.01, \kappa_5=1.17, \kappa_6=1.92, \kappa_7=1.11$ and $\kappa_8=1.88$. The parameters are sampled uniformly randomly in $[1,2]$. We use $\mu=5$ and $\theta=1$ for the ACR controller. \textbf{b.} Distribution of $X$ at $t=5$ in the uncontrolled system in Figure \ref{fig:faststochastic}a with both the chosen parameters and randomly perturbed parameters. \textbf{c.} Distribution of $X$ at $t=5$ in the controlled system using both the chosen parameters and randomly perturbed parameters. Red dotted line indicates the Poisson distribution with the rate $5$. \textbf{d.} The controlled system is approximated with a hybrid model consisting of stochastic and deterministic parts. \textbf{e.} The mean $m(t)$ of the distribution of $X$ is displayed (blue solid line) along with the distribution of X at times $t=0.2$ and $t=5$.}}\label{fig:faststochastic} 
\end{figure}
%

\section{Discussion}

Absolutely robust networks have the property that the steady state value of a target species is independent of the total mass of the system. In this paper we have provided a class of controllers based on absolutely robust networks. We define a control species that interacts with the target species, embedding an absolutely robust network into the given network to enforce target species robustness. For deterministically modeled networks, this type of controller not only stabilizes the target species at the desired value by tuning the parameters of the ACR controller, but also makes the species robustly adapted to parameter perturbations and a transiently supplied additional reactions. We demonstrate control for deterministic system through an ACR system with an ERK model.  We illustrate some of our results with the so-called base ACR controller, but we show that other ACR networks can be also be used.

We also show that ACR controllers have the ability to control stochastic networks. The need for control stochastic system is becoming clear in many disciplines of systems and synthetic biology, particularly given the low species counts present in many individual cells.  The average of a species concentration is a deterministic quantity of a stochastic system, thus one might think that a controller used for a typical deterministic system could also implement stochastic control. However for a nonlinear system, studying the dynamics associated with the averages requires nontrivial tools such as moment closure \cite{gillespie2009moment}. Even if mean control is valid with a given controller, the system may be still out of control if noise is not properly accounted for. Furthermore, because the associated stochastic system describes molecular counts of each species instead of concentrations, some species might reach a zero state and lead to an extinction event.  
   
As a result, for the control of stochastic systems it can be helpful to use advanced mathematical tools such as theoretical analysis of chemical reaction networks. Using an ACR controller for stochastic systems here involves two main mathematical tools, multiscaling model reduction and deficiency zero theorems. To avoid a potential breakdown of a controller because of lack of  reactants, we design an ACR controller with high copies of the control species. Using the tools above, we show that a species of interest in the controlled system is roughly Poissonian with tunable mean and variance. Combining the multiscaling model reduction and the zero deficiency condition, we show that a simple ACR system can control both mean and variance of an inactive receptor in stochastic receptor-ligand system as the distribution of the inactive receptor roughly follows a Poisson distribution centered at the desired value. The controlled stochastic system also admits robust perfect adaptation as does the corresponding deterministic system.

We note that the basic ACR controller used throughout this paper has a connection to classical control theory, as it admits a non-linear integral feedback that is a well-studied characteristic of robustly adapted systems \cite{ferrell2016perfect, ma2009defining, yi2000robust}. Integral feedback loops arise in many important biological phenomena such as bacterial chemotaxis, photoreceptor responses, or MAP kinase activities. For the simple mass-action ACR system
\begin{align}
Z+X\xrightarrow{\theta} 2Z, \quad Z\xrightarrow{\mu} X,
\end{align}
 the concentration of $Z$ satisfies
$\quad \frac{d}{dt}z(t)=z(t)(-\mu+\theta x(t))$.
Dividing by $z(t)$ and integrating on both sides, we obtain 
\begin{align}
\log z(t)=\log z(0)+\int_0^t (\theta x(s)-\mu)ds,
\end{align}
which is a non-linear integral feedback relation. Such types of integral feedback loops appear in many different biochemical systems \cite{briat2016antithetic, briat2018antithetic, cappelletti2019hidden, shoval2010fold, xiao2018robust}.

One of the major issues on synthetic controllers is the practical implementation of the proposed controller. Aoki et al. \cite{aoki2019universal} show that an antithetic controller could be constructed using two control proteins, $\sigma$ factor SigW and anti-$\sigma$ factor RsiW, in an emph{E. coli} plasmid implementation. 
For an ACR controller, it remains an open question whether its design is practically feasible \emph{in vivo} or \emph{in vitro}. One key for synthesizing it is the bifunctionality of an enzyme that potentially brings ACR to the system, as it has been observed for other ACR applications \cite{shinar2007input, dexter2013dimerization}. Notice that the control species $Z$ mediates both production and degradation of the target species $X$ in the basic ACR controller. We have suggested some ideas for implementing ACR controllers in our examples. The control species could be obtained by phosphorlyation of a bifunctional target species, by antagonist ligand binding, or by a form of protein misfolding. 

As sufficient network architectural conditions for ACR property have been shown for example in the regulation of osmolarity in bacteria \cite{shinar2010structural}, designing more general ACR controllers could be feasible. Therefore we believe that this new approach introduced in this paper could help control other biochemical networks in a way that takes into account stochastic effects.
\section*{Acknowledgements}
We would like to thank Eduardo Sontag, Carsten Wiuf, Chuang Xu, Linard Hoessly and Jason Dark for key suggestions regarding this work.

\section*{Funding}
This work is partially supported by NSF grant DMS1763272 and Simons Foundation grant 594598 (Qing Nie). This work is also supported by NSF grant DMS1616233.

\newpage

\begin{center}
\Huge{Supplementary Material}
\end{center}

\section{Chemical reaction network theory}

\subsection{Reaction networks}
\label{sec:reactionnetworks}

In this section, we provide mathematical models associated with biochemical systems that we use in the main manuscript, starting with the introduction of reaction networks. A biochemical system can be described with a reaction network, which consists of constituent species, complexes that are combinations of species, and reactions between complexes. A triple $(\S,\C,\Re)$ represents a reaction network where $\S, \C$ and $\Re$ are collections of species, complexes and reactions, respectively.

\begin{example}
Consider the following reaction network describing a substrate-enzyme system.
\[
	S+E \rightleftharpoons SE \rightarrow E+P,
\]
For this reaction network, $\S = \{S,E,SE,P\}$, $\C=\{S+E,SE,E+P\}$ and $\Re=\{S+E\rightarrow SE, SE\rightarrow S+E,SE\rightarrow E+P\}$. \hfill $\triangle$
\end{example}

Regarding a reaction network as a directed graph, each connected component is termed a \emph{linkage class}. A subset $Q$ of complexes in a linkage class is a \emph{strongly connected component} if and only if for any two complexes $y, y'\in Q$, there exists a path of directed edges connecting from $y$ to $y'$. 
If every linkage class in a network consists of a single strongly connected component, then the network is \emph{weakly reversible}. By the definition, in a network $(\S,\C,\Re)$, the set of complexes $\C$ can be decomposed into disjoint linkage classes. Allowing that a single complex can be a strongly connected component, every linkage class is decomposed into disjoint strongly connected components.

For example, for the following network $(\S,\C,\Re)$
\begin{equation}\label{eq:example components}
\emptyset \rightleftharpoons C, \quad A \rightleftharpoons B \to A+B \to 2C\rightleftharpoons B+C,
\end{equation}
there are two linkage classes $\{\emptyset,C\}$ and $\{A,B,A+B,2C,B+C\}$. Linkage class $\{\emptyset,C\}$ consists of a single strongly connected component. Linkage class $\{A,B,A+B,2C,B+C\}$ has three strongly connected components $\{A,B\}, \{A+B\}$ and $\{2C,B+C\}$. 

Each strongly connected component is further classified into two categories. For a strongly connected component $Q$, if there is no path of directed edges connecting from $y\in Q$ to $y'\not \in Q$, then $Q$ is a \textit{terminal connected component}. Otherwise, $Q$ is a \textit{non-terminal connected component}.
A complex contained in a terminal connected component is called a \textit{terminal complex}, otherwise it is called a \textit{non-terminal complex}. In \eqref{eq:example components}, strongly connected components $\{\emptyset,C\}$ and $\{2C,B+C\}$ are terminal connected components, and the others are non-terminal connected components.

We introduce a domain on which the dynamical system associated with a reaction network is defined.
\begin{defn}
Let $(\S,\C,\Re)$ be a reaction network. 
For a $x_0\in \R^d_{>0}$, we call a set $S_{x_0}=x_0+\text{span}\{y'-y:y\to y' \in \Re\}\cap \R^d_{>0}$ the \textit{stoichiometry class}.
\end{defn}

\subsection{Dynamical systems}

For a dynamical system of a reaction network, a reaction rate constant $\kappa$ for each reaction $y\to y'$ gives a weight on each reaction, and we denote $y\xrightarrow{\kappa} y'$ to incorporate the rate constant. With a collection of rate constants $\K$, we denote the associated dynamical system for $(\S,\C,\Re)$ by $(\S,\C,\Re,\K)$. 

%
%
%

For mathematical models of reaction networks, we typically assume that the associated system is spatially well-stirred. In this case the usual mathematical model for a reaction network is either a system of ordinary differential equation or a continuous-time, discrete-space Markov process. When each species has high copy number so that intrinsic noise can be averaged out, the concentration vector $x(t)$ of species in a reaction network $(\S,\C,\Re)$ is typically modeled with a deterministic network system  
\begin{equation}\label{eq:deterministic system}
\dfrac{d}{dt}x(t) = \sum_{y\rightarrow y'}\kappa_{y\rightarrow y'} \eta_y (x(t))(y'-y), 
\end{equation}
where $\eta_y : \R^d_{>0}\to \R_{> 0}$ is a rate function associated to a reaction $y\xrightarrow{\kappa_{y\to y'}} y'$.
One of the prevalent choice of the rate function is \textit{mass action kinetics} which defines $\eta_{y}(x)=x^y$, where $u^v = \prod_{i=1}^d u_i^{v_i}$ for two vectors $u,v\in \R^d_{\ge 0}$.

The intrinsic stochasticity of a system is considered when each species in a reaction network system has low copy number. For the usual stochastic model, we use a continuous time, discrete state space Markov process  $X(t) \in \Z^d_{\ge 0}$ defined on $\Z^d_{\ge 0}=\{z \in \Z^d:z_i \ge 0 \text{ for each } i\}$. The transitions of $X$ are determined by the reaction vectors. Letting $h(t)$ be a function such that $\dlim_{t\to 0}\dfrac{h(t)}{t}=0$, the transition probabilities are defined as
\begin{equation}\label{eq:prob}
P(X(t+\Delta t) = z+y'-y \ | \ X(t)=z) = \sum_{\substack{\bar y \to \bar y'\\ \bar y'-\bar y =y'-y}}\kappa_{\bar y \to \bar y'}\lambda_{\bar y\rightarrow \bar y'}(z)  \Delta t + h(\Delta t), 
\end{equation} where $y'-y$ is a reaction vector associated with a reaction $y\to y'$, and $\lambda_{y\to y'}:\Z^d_{\ge 0} \to \R_{\ge 0}$ is the \emph{reaction intensity} representing how likely the associated reaction $y\to y'$ fires. 


The usual choice of the propensity functions for a stochastic network system $(\S,\C,\Re)$ is
\begin{equation}\label{mass}
\lambda_{y\rightarrow y'}(x)= x^{(y)},
\end{equation} 
where $u^{(v)}=\displaystyle \prod_{i=1}^d \dfrac{u_i!}{(u_i-v_i)!}\mathbf{1}_{\{u_i \ge v_i\}}$ for $u,v \in \Z^d_{\ge 0}$. This choice of the propensity function is \textit{stochastic mass-action kinetics}. In both this supplementary material and the main text, we model both the deterministic and the stochastic dynamical system under mass action kinetics. Letting $\K$ be the set of reaction intensities associated with $\Re$, the quadruple $(\S,\C,\Re,\K)$ defines a (either deterministic or stochastic) dynamical system associated with the reaction network $(\S,\C,\Re)$.

An infinitesimal behavior of the associated $X$ can be described with the infinitesimal generator $\mathcal{A}$  \cite{Kurtz86},
\begin{equation}\label{gen5}
 \mathcal{A}V(x) = \lim_{h\to 0} \frac{E_x(V(X(h)))-V(x)}{h}=\sum_{y \rightarrow y' \in \Re} \lambda_{y\rightarrow y'}(x)(V(x+y'-y)-V(x)),
\end{equation}
for a function $V:\Z^d_{\ge 0} \to \R$, where $E_x$ denotes the expectation of the process whose initial point is $x$.



\subsection{Deficiency zero theory}\label{subsec:deficiency}
The deficiency of a reaction network is a positive integer determined solely by the structure of the network regardless of parameter values. Let $(\S,\C,\Re)$ be a reaction network with $m$ complexes and $\ell$ linkage classes. Let further $s$ be the rank of the stoichiometric matrix whose $i$-th column is given by $i$-th reaction in $\Re$. The deficiency $\delta$ is equal to
\begin{align*}
m-\ell-s.
\end{align*}
There are a couple of interpretations of the deficiency. First, we can represent the deterministic system \eqref{eq:deterministic system} as
\begin{align*}
\frac{d}{dt}x(t)=YA_\K \psi(x(t))
\end{align*}
with a stoichiometry coefficient matrix $Y$, rate constant matrix $A_\K$ and the rate function $\phi(x)$ (See \cite{FeinbergLec79} for more details). Then the deficiency $\delta$ of a network $(\S,\C,\Re)$ satisfies
\begin{align*}
\delta=dim(Ker(Y)\cap Im(A_\K)).
\end{align*} 
Second, the deficiency roughly stands for redundancy of the network in the following sense. Consider the following two networks,
\begin{align*}
\emptyset \rightleftharpoons A \quad \text{and} \quad \emptyset \rightleftharpoons A \rightleftharpoons  2A.
\end{align*} 
The deficiency of the left reaction network is $0=2-1-1$. The deficiency of the right reaction network is $1=3-1-1$. This difference stems from the additional reaction $A \rightleftharpoons 2A$ in the right network. The gain and loss of one $A$ species is already realized with reaction $\emptyset \rightleftharpoons A$. Hence reaction $A \rightleftharpoons 2A$ is redundant. 


Zero deficiency combined with weak reversibility of reaction networks implies very strong characteristics of the associated system dynamics for both deterministic models and stochastic models. 
\begin{thm}[Horn 1972 \cite{Horn72}, Feinberg 1972 \cite{Feinberg72}]
Let $(\S,\C,\Re)$ be a weakly reversible reaction network with zero deficiency. Then for any choice of rate parameters, the associated deterministic dynamics endowed with the mass-action kinetics admits a unique locally asymptotic stable positive steady state at each stoichiometry class.
\end{thm}

The stationary distribution of the associated stochastic process is fully characterized for a weakly reversible network which has zero deficiency.
\begin{thm}[Anderson, Craciun and Kurtz 2010 \cite{AndProdForm}]\label{thm:prod}
Let $(\S,\C,\Re)$ be a weakly reversible reaction network with zero deficiency. Then for any choice of rate parameters, the associated Markov process endowed with the stochastic mass-action kinetics admits a stationary distribution, and it is a product form of Poissons (or constrained Poissons). That is, for each $x \in Z^d_{\ge 0}$ in the state space, the stationary distribution $\pi$ satisfies
\begin{align*}
\pi(x)=M\prod_{i=1}^d \frac{c_i^{x_i}}{x_i!}
\end{align*}
where $c=(c_1,c_2,\dots,c_d)$ is a steady state of the deterministic counterpart and $M$ is the normalizing constant.
\end{thm}

For control of a stochastic model, we use Theorem \ref{thm:prod} to find an approximation of a target species in a controlled system. Details about this procedure is state in Section \ref{sec:control stochastic}.

\section{ACR systems and ACR controllers}\label{sec:acr controllers}
In this section we introduce the absolute concentration robustness (ACR) of a reaction network. 
In order to make use of ACR systems to design a controller, we consider a special class of ACR networks, and then we introduce a precise definition of an ACR controller.
 
\begin{defn}
Let $\hat x$ be a solution to the deterministic network system $(\hat \S,\hat \C,\hat \Re,\hat \K)$ such that
\begin{align*}
\frac{d}{dt}\hat x(t)=\hat f(\hat x(t)).
\end{align*} Suppose this system admits a positive steady state. If there exists a species $X_1\in \hat \S$ such that the values of $X_1$ at any positive steady states are all identical, then $(\hat \S,\hat \C,\hat \Re,\hat \K)$ is called an \textit{ACR network system}. Furthermore, the species $X_1$ and the identical positive steady state value of $X_1$ are called an ACR species and an ACR value, respectively. Especially if the deterministic model is equipped with mass-action kinetics, the system is called a \textit{mass-action ACR network system}. 
\end{defn}

In some special cases, the ACR property is determined with a single species in a network system.
\begin{defn}\label{def:definite acr}
Let $(\hat \S,\hat \C,\hat \Re,\hat \K)$ be a deterministic network system modeled with
\begin{align*}
\frac{d}{dt}\hat x(t)=\hat f(\hat x(t))
\end{align*}
If there exists species $S_i \in \hat \S$ such that $\{\hat x_1:\hat f_i(x')=0, \hat x=(\hat x_1,\dots,\hat x_d) \in \Re^{d}_{>0}\}=\{c\}$ for some $c>0$, then the deterministic system is termed \textit{a $S_i$-definite ACR system}. 
\end{defn}

\begin{rmk}
An $S_i$-definite ACR system is an ACR system.
For a $S_i$-definite ACR system, an ACR species and its ACR value is solely determined by the single equation associated with the species $S_i$.
\end{rmk}

A simple mass-action ACR system constructed with only two species is introduced in \cite{shinar2010structural}. 
Let $(\hat \S,\hat \C,\hat \Re,\hat \K)$  be mass-action system associated with 
\begin{equation}\label{eq:basic acr}
Z+X_1\xrightarrow{\theta} 2Z, \quad Z\xrightarrow{\mu} X_1.
\end{equation}
Because any positive roots $x^*=(x_1^*,z^*)$ of the equation $\frac{d}{dt}z(t)=z(t)(\theta x_1(t)-\mu)$ for species $Z$ satisfies $x^*_1=\frac{\mu}{\theta}$, this mass action system is an ACR system. Furthermore since we have $\{x_1 : z(\theta x_1-\mu)=0\}=\{\frac{\mu}{\theta}, x_1>0, z>0\}$, this system is also a $Z$-definite ACR system by Definition \ref{def:definite acr}. We termed this system a \textit{basic ACR system for $X_1$}. This ACR system would be mainly used for control in the main text. 

It is shown that there is a broad collection of networks whose associated mass-action system are ACR systems. They are characterized using network topological conditions in \cite{shinar2010structural}. 
In the following theorem, $e_i$ denotes a vector whose $i$ th entry is one, and the other entries are all zeros.
\begin{thm}[Shinar and Feinberg 2010 \cite{shinar2010structural}]\label{thm:acr}
Let $(\S,\C,\Re)$ be a deficiency 1 reaction network.
Suppose there are two non-terminal complexes $y$ and $\bar y$ such that $y-\bar y=ce_i$ for some $i\in Z_{>0}$ and $c\neq 0$. Then for any set of parameters $\K$, the mass-action deterministic network system $(\S,\C,\Re,\K)$ is a mass-action ACR network system.
\end{thm}

The controlled deterministic system is basically a union of two deterministic systems; one is a given network system and the other is an ACR system. We formally define the union of two deterministic network systems.
In the definition below, $M_{n,m}$ denote the set of all $n\times m$ matrices and $I_n$ denotes the $n\times n$ identity matrix.

\begin{defn}
Let $(\S,C\,\Re,\K)$ and $(\hat \S,\hat \C,\hat \Re,\hat \K)$ be deterministic network systems modeled with
\begin{align*}
\frac{d}{dt}x(t)=f(x(t))\quad \text{and} \quad \frac{d}{dt}\hat x(t)=\hat f(\hat x(t)), \text{ respectively.}
\end{align*}
Let $\S=\{X_1,\dots,X_d,Y_1,\dots,Y_k\}$ and $\hat \S=\{X_1,\dots,X_{d},Z_1,\dots,Z_{\hat k}\}$. Then the union system of the deterministic systems $(\S,C\,\Re,\K)$ and $(\hat \S,\hat \C,\hat \Re,\hat \K)$ is a deterministic system such that
\begin{align*}
\frac{d}{dt}\bar x(t)=\bar{f} (\bar x(t)),
\end{align*}
where $\bar x(t)=(x_1(t),\dots,x_d(t),y_1(t),\dots,y_k(t),z_1(t),\dots,z_{\hat k}(t))^T$ and $\bar f= E f+\hat E \hat f$ with

\begin{align*}
E=\begin{pmatrix}
I_{d} & 0\\
0 & I_k\\
0 & 0
\end{pmatrix} \in M_{d+k+\hat k,d+k}, \quad \text{and} \quad
 \hat E=\begin{pmatrix}
I_d & 0\\
0& 0\\
0& I_{\hat k}
\end{pmatrix} \in M_{d+k+\hat k,d+\hat k}.
\end{align*}
\end{defn}

Now, we define an ACR controller.

\begin{defn}
Let $(\S,\C,\Re,\K)$ be a deterministic network system and let $(\hat \S,\hat \C,\hat \Re,\hat \K)$ be an ACR network system such that $X_1 \in \S\cap \hat \S$ and $\hat \S\setminus \S \neq \emptyset$. If the union of the two network systems is an ACR network system such that $X_1$ is an ACR species, then $(\hat \S,\hat \C,\hat \Re,\hat \K)$ is termed an ACR controller for $(\S,\C,\Re,\K)$ and the union system is called a \textit{controlled system}. Furthermore, if ACR controller $(\hat \S,\hat \C,\hat \Re,\hat \K)$ for $(\S,\C,\Re,\K)$ is a mass-action system, then it is termed a \textit{mass-action ACR controller} for $(\S,\C,\Re,\K)$. 
\end{defn}


\section{Steady states and stability using an ACR controller}\label{sec:stability}
In this section, we show that for any deterministic system modeled with general kinetics, an ACR controller endows ACR to the given system and drives the long-term behavior of a target species towards the desired value. For the basic ACR controller, the existence of the steady states will now be verified together with their stability. In this manuscript, every chemical reaction network is modeled under mass action kinetics.  It is notable, however, that Lemma \ref{thm:x1 is acr if ps exists}, Lemma \ref{thm:ps exists with basic acr}, and Theorem \ref{thm:stability} hold for a class of non-mass action networks, for instance under the generalized kinetics defined in \cite{FeinbergLec79} (Definition 2.2).

\begin{lem}\label{thm:x1 is acr if ps exists}
Let $(\S,\C,\Re,\K)$ be a deterministic network system such that $X_1 \in \S$. Let $(\hat \S,\hat \C, \hat \Re, \hat \K)$  be a $Z$-definite ACR system such that $X_1\in \hat \S$ and $Z\not \in \S$. If the union system of $(\S,\C,\Re,\K)$ and $(\hat \S,\hat \C, \hat \Re, \hat \K)$ admits a positive steady state, then it is an ACR system, and $X_1$ is an ACR species.
\end{lem}
\begin{proof}
Since $Z\not \in S$, the equation $\frac{d}{dt}z(t)=0$ for $Z$ in the union system is same as the equation for $Z$ in ACR system $(\hat S, \hat \C, \hat \Re, \hat \K)$.  At the expense of abusing the notation, we let $\frac{d}{dt}z(t)=f_z(x(t))$ and $\frac{d}{dt}\bar z(t)=f_{z}(\bar x(t))$ be the equations for $Z$ in $(\hat \S,\hat \C, \hat \Re, \hat \K)$ and the union system, respectively. By definition of the $Z$-definite ACR system,
if $f_z(x)=0$ then there exists a positive real number $c$ such that $x_1=c$. Hence, for each positive steady state $\bar x^*$ in the union system, $\bar x^*_1=c$ and therefore $X_1$ is an ACR species in the union system. 
\end{proof}

For a given network system $(\S,\C,\Re)$, suppose $Z\not \in \S$ and $X_1 \in \S$. Since the basic ACR controller \eqref{eq:basic acr} for $X_1$ is a $Z$-definite ACR system, it is a mass action ACR controller for $(\S,\C,\Re)$. We call this basic ACR system the \emph{basic ACR controller} interchangeably.

Lemma \ref{thm:x1 is acr if ps exists} guarantees that the values of $X_1$ must be $c$ at any positive steady states as long as a positive steady state exists in the union system. The following \textcolor{black}{lemma} provides a sufficient condition of a given network system $(\S,\C,\Re,\K)$ for existence of a positive steady state in the union system of  $(\S,\C,\Re,\K)$ and the basic ACR controller. \textcolor{black}{We show that if the desired control value is within the range of the observations, then one can control for that value. }

\begin{lem}\label{thm:ps exists with basic acr}
Let $(\S,\C,\Re,\K)$ be a deterministic network system modeled with
\begin{align*}
\frac{d}{dt}x(t)=f(x(t)).
\end{align*} Let $PS=\{x : x=(x_1,x_2,\dots,x_d)\in \Re^d_{>0}, f(x)=0 \}$. Suppose there exists an $x^* \in PS$ such that $x^*_1=\frac{\mu}{\theta}$, then the basic ACR network system \eqref{eq:basic acr} is a mass-action ACR controller for $(\S,\C,\Re,\K)$, and the controlled system admits a positive steady state.
\end{lem} 
In the following proof, the concatenation $w=(u,v)$ for $u\in \R^d$ and $b\in \R^1$ denotes a vector in $\R^{d+1}$ such that $w_i=u_i$ for $i=1,2,\dots,d$ and $w_{d+1}=v$.
\begin{proof}
Let $\S=\{X_1,\dots,X_d\}$. Let $\bar x=(x,z)$ for each $x\in \R^d_{\ge 0}$ and $z \in \R_{\ge 0}$ be a solution to the union system of $(\S,\C,\Re,\K)$ and the basic ACR network system \eqref{eq:basic acr}. Then $\bar x$ satisfies
\begin{equation}
\frac{d}{dt}\bar x(t)=\bar f (\bar x(t)),
\end{equation}
for some $\bar f$. 
 By the construction of the union system, we have 
\begin{equation}\label{eq:controlled system equations}
\begin{split}
\bar f_i(\bar x)=\begin{cases} f_i(x)-z(\theta x_1-\mu)\quad &\text{if $i=1$},\\
z(\theta x_1-\mu) \quad &\text{if $i=d+1$},\\
f_i(x), \quad &\text{otherwise}.
\end{cases}
\end{split}
\end{equation}
Let $x^*$ be a positive steady state of $(\S,\C,\Re,\K)$ such that $x_1=\frac{\mu}{\theta}$. For any positive value $z^*$, we have $\bar f(\bar x^*)=0$ where $\bar x^*=(x^*,z^*)$. 
\end{proof}

\textcolor{black}{For a given choice of $\theta$ and $\mu$, it is a necessary condition that there exists $x^*\in PS$ such that $x^*_1=\dfrac{\mu}{\theta}$ for a given reaction system because unless $Z=0$, the control species $Z$ is only stabilized when $X_1=\dfrac{\mu}{\theta}$ . We demonstrate this with the following example.}

\begin{example}
\textcolor{black}{Consider this system suggested by one of the reviewers of this manuscript: 
\begin{align*}
 &A  \xrightarrow{\ 1 \ } B\\[-1ex]
\text{\scriptsize{$1$}}&\nwarrow \ \ \ \swarrow \text{\scriptsize{$1$}} \\[-1ex]
& \ \ \ \ \ \ 0
\end{align*}
Note that this system admits a unique equilibrium $(1,1)$. Hence for the choice of $\theta=1$ and $\mu=2$, there is no $x^*$ in $PS$ such that $x^*_1=\dfrac{\mu}{\theta}$. Notably the union of the system and the basic ACR system $A+Z\xrightarrow{\theta} 2Z$ and $Z\xrightarrow{\mu} A$ does not admit a positive steady because no positive values $a^*$ and $z^*$ satisfy both 
\begin{align*}
0=-a^*+1+z^*(2-a^*), \quad \text{and} \quad 0=z^*(a^*-2).
\end{align*}}
\end{example}

The convergence to positive steady states in general controlled systems with a ACR controller is more delicate problem since the actual network structure and parameters need probably to be specified. However, if linear stability condition is held for a given system as well as the conditions in Lemma~\ref{thm:ps exists with basic acr} with some additional conditions, then the controlled system with the basic ACR system \eqref{eq:basic acr} admits linear stability.
Linear stability of a steady state holds if each eigenvalues of the Jacobian of a dynamical system at the steady state has a strictly negative real part. This implies the dynamical system  asymptotically converges to the steady state if its initial state was close enough to the steady state. 

Remark that in case a given system has no conservation relation, the dynamics is not confined into a lower dimensional stoichiometry class. Hence if we assume linear stability of the given system at a positive steady state $x^*$, all eigenvalues of the Jacobian at the steady state have strictly negative real parts. Hence we can maintain the linear stability after we add a ACR controller if the parameters of the ACR controller are small enough. This is by the fact that the roots of the characteristic polynomial are continuous with respect to the coefficients, hence the eigenvalues of the Jacobian of the controlled system still have strictly negative real parts.

 Hence we investigate the stability of the controlled system when a given system $(\S,\C,\Re,\K)$ admits conservation relations.
Let $x(t)=(x_1(t),\dots,x_d(t))$ be the deterministic model associated with $(\S,\C,\Re,\K)$ such as \eqref{eq:deterministic system} in $\R^d_{>0}$. Suppose that $u^1,\dots,u^k$ are positive vectors such that $u^i \cdot \frac{d}{dt}x(t)=0$ for all $t$ and for each $i$, where $\cdot$ means the canonical inner product between two finite dimensional euclidean vectors. This implies that for a fixed initial state $x(0)$, there exist $M_i$'s such that
\begin{align*}
u^i\cdot x(t) = M_i \quad \text{for all $t$}.
\end{align*}

Without loss of generality, we suppose $u^i$ are linear independent. Then in the following way, we can reduce the system onto a lower dimension system that admits no conservative relations. First note that since we assume the linear independence of $u^i$'s, we have  $k\le d$. Hence using  Gaussian elimination and by rearranging the coordinate of $x$, we have
\begin{equation}\label{eq:original conserve}
\begin{bmatrix}
U | I
\end{bmatrix}
\begin{bmatrix}
x_1(t)\\x_2(t)\\ \vdots \\x_{\bar d}(t)\\ x_{\bar d+1}(t) \\x_d(t)
\end{bmatrix}=
\begin{bmatrix}
M_1\\  M_2\\\vdots \\  M_k
\end{bmatrix},
\end{equation}\label{eq:conv for stability}
where $\bar d= d-k$, the matrix $I$ is the $k$ dimensional identity matrix, $U$ is some $\bar d \times k$ matrix and $ M_i$'s are some constants. Hence we have 
\begin{equation}\label{eq:reduce conserve}
\begin{aligned}
x_{\bar d + 1}(t)&=M_1-\sum_{i=1}^{\bar d} u_{1i}x_i(t),\\
x_{\bar d + 2}(t)&=M_2-\sum_{i=1}^{\bar d} u_{1i}x_i(t),\\
&\vdots\\
x_{d}(t)&=M_k-\sum_{i=1}^{\bar d} u_{1i}x_i(t).\\
\end{aligned}
\end{equation}
This implies the variables $x_{\bar d+1},\dots,x_{d}$ are completely determined by relations \eqref{eq:reduce conserve}. Then we have the following reduced system,
\begin{equation}
\begin{aligned}
&\frac{d}{dt}x_i(t)=g_i(x_1(t),x_2(t),\dots,x_{\bar d}(t)) \quad \text{where},\\
&g_i(x_1,x_2,\dots,x_{\bar d})=f_i\left (x_1,x_2,\dots,x_{\bar d},M_1-\sum_{i=1}^{\bar d} u_{1i}x_i(t),\dots,M_k-\sum_{i=1}^{\bar d} u_{ki}x_i(t) \right ).
\end{aligned}
\end{equation}
for $i=1,2,\dots,\bar d$. Note that this reduced system is specified with the choice of initial state $x(0)$ as the initial condition determines the conservative quantity $M_i$'s.
Note further that since the steady state values of $x_{\bar d+i}$ for $i=1,2\dots,k$ are completely determined by the steady state values $x_i$ for $i=1,2,\dots,\bar d$, the stability of $x(t)$ is also determined by the reduced system $(x_1(t),\dots,x_{\bar d}(t))$. The linear stability of the reduced system is investigated with the eigenvalues of Jacobian. 
We denote $J(x^*)$ be the Jacobian of this reduced system at $x^*$, where we abuse the notation since $x^*$ is a state in the original system but the Jacobian is for the reduced system.

Now we suppose that species $X_1$ is the control target with the basic ACR controller 
\begin{equation}\label{eq:basic acr for stability}
X_1+Z \xrightarrow{\theta} 2Z, \quad Z \xrightarrow{\mu} X_1.
\end{equation}
Suppose that $S_1$ is involved in at least one conservation relation. Without loss of generality, suppose $u^1_{1}=1$.  Then we have new conservation relations in the union system of $(\S,\C,\Re,\K)$ and the basic ACR system. We let
\begin{equation}\label{eq:new conv}
\bar M_i = 
\begin{cases}
u^i\cdot x(0) + u^i_1z(0)=M_i+z(0) \quad &\text{if $u^i_1\neq 0$},\\
u^i\cdot x(0) + u^i_1z(0)=M_i, &\text{otherwise}.
\end{cases}
\end{equation}
 
Then the new conservative relations are represented as
\begin{equation}\label{eq:conv stability of union}
\begin{bmatrix}
U \ | \ u_1 \ | \ I
\end{bmatrix}
\begin{bmatrix}
x_1(t)\\x_2(t)\\ \vdots \\x_{\bar d}(t) \\ z(t) \\ x_{\bar d+1}(t)  \\x_d(t)
\end{bmatrix}=
\begin{bmatrix}
\bar M_1\\ \bar M_2\\\vdots \\ \bar M_k
\end{bmatrix},
\end{equation}
where $v$ is the first column vector of $U$, and $U$ and $I$ are the same matrices as \eqref{eq:original conserve}. The definition of $u_1$ basically means that the control species $Z$ is involved in the same conservation relation as $X_1$ in the original system $(\S,\C,\Re,\K)$. Hence the dynamics $\bar x(t)=(\bar x_1(t),\dots,\bar x_d(t),z(t))$ associated with the union system can also be reduced to
\begin{equation}\label{eq:reduced with z}
\frac{d}{dt}\bar x_i(t)=h_i(x_1(t),x_2(t),\dots,x_{\bar d}(t),z(t)).
\end{equation}
where,
\begin{align*}
&h_i(x_1,x_2,\dots,x_{\bar d},z)=\\
&\begin{cases} 
&f_1\left (x_1,x_2,\dots,x_{\bar d},\bar M_1-\dsum_i^{j=\bar d} u^1_{j}x_j(t)-z(t),\dots,\bar M_k-\dsum_{j=1}^{\bar d} u^k_{j}x_j(t)-z(t)\right ) -z(\theta x_1-\mu), \text{if $i=1$},\\
&z(\theta x_1-\mu) \hfill \text{if $i=\bar d+1$},\\
&f_i\left (x_1,x_2,\dots,x_{\bar d},\bar M_1-\dsum_i^{i=\bar d} u_{1i}x_i(t)-v_1z(t),\dots,\bar M_k-\dsum_{j=1}^{\bar d} u^i_{j}x_j(t)-u^i_1z(t) \right ),\hfill \text{otherwise}\\
\end{cases}
\end{align*}

Note that by \eqref{eq:new conv}, we have $\partial_i h_j(x^*_1,\dots,x^*_{\bar d})=\partial_i g_j(x^*_1,\dots,x^*_{\bar d})$ for $i,j \in \{1,2,\dots,\bar d\}$. 
Then the jacobian $\bar J(x^*,z^*)$ for this system at $(x^*,z^*)$ where $x^*_1=\dfrac{\mu}{\theta}$ is 
\begin{equation}\label{eq:jacobian reduced}
\begin{split}
\bar J(x^*,z^*)=&\begin{bmatrix}
\partial_1 h_1(x^*,z^*)-\theta z^*& \partial_2 h_1(x^*,z^*) & \cdots &\partial_{\bar d} h_{1}(x^*,z^*)&  \partial_{z} h_{1}(x^*,z^*)\\
\partial_1 h_2(x^*,z^*) &\partial_2 h_2(x^*,z^*) & \cdots & \partial_{\bar d} h_{2}(x^*,z^*)&\partial_{z} h_{2}(x^*,z^*)\\
& & \vdots & \\
\partial_1 h_{\bar d}(x^*,z^*) &\partial_2 h_{\bar d}(x^*,z^*) & \cdots & \partial_{\bar d} h_{\bar d}(x^*,z^*)& \partial_{z} h_{\bar d}(x^*,z^*)\\
\theta z^* & 0 &\cdots &0 & 0
\end{bmatrix}\\
&=\begin{bmatrix}
\partial_1 g_1(x^*)-\theta z^*& \partial_2 g_1(x^*) & \cdots &\partial_{\bar d} g_{1}(x^*)&  \partial_{z} h_{1}(x^*,z^*)\\
\partial_1 g_2(x^*) &\partial_2 g_2(x^*) & \cdots & \partial_{\bar d} g_{2}(x^*)&\partial_{z} h_{2}(x^*,z^*)\\
& & \vdots & \\
\partial_1 g_{\bar d}(x^*) &\partial_2 g_{\bar d}(x^*) & \cdots & \partial_{\bar d} g_{\bar d}(x^*)& \partial_{z} h_{\bar d}(x^*,z^*)\\
\theta z^* & 0 &\cdots &0 & 0
\end{bmatrix}
\end{split}
\end{equation}

As the stability of $x(t)$ is determined by the stability of its reduced system, we consider the linear stability of the reduced system \eqref{eq:reduced with z} of $\bar x(t)$ to study the stability of the union system of $(\S,\C,\Re,\K)$ and the basic ACR controller \eqref{eq:basic acr for stability}.
For the linear stability, we show that the characteristic function for $\bar J(x^*,z^*)$ will be  the characteristic function of $J(x^*)$ with some perturbation. 
To show this, we use the conventional notation for deterministic $|A|$ for a square matrix $A$. $I$ denotes an identity matrix, and it could denote a different dimensional identity matrix according to the content. We will also use the column/row expansion of deterministic. We further also use the row decomposition for determinant. The row decomposition of the determinant means that when $A$ is a square matrix such that the first row $A_1$ is equal to $A'_1+A''_1$ with some row vectors $A'_1$ and $A''_1$, we have $|A|=|A'|+|A''|$ where $A'$ and $A''$ are square matrices whose first row is replaced with $A'_1$ and $A''_1$, respectively.

Then
\begin{equation}\label{eq:dets}
\begin{aligned}
&|\lambda I- \bar J(x^*,z^*)|=\\
&\begin{vmatrix}
\lambda-\partial_1 g_1(x^*)+\theta z^*& -\partial_2 g_1(x^*) & \cdots &-\partial_{\bar d} g_{1}(x^*)&  -\partial_{z} h_{1}(x^*,z^*)\\
-\partial_1 g_2(x^*) &\lambda-\partial_2 g_2(x^*) & \cdots & -\partial_{\bar d} g_{2}(x^*)& -\partial_{z} h_{2}(x^*,z^*)\\
& & \vdots & \\
-\partial_1 g_{\bar d}(x^*) &-\partial_2 g_{\bar d}(x^*) & \cdots & \lambda-\partial_{\bar d} g_{\bar d}(x^*)& -\partial_{z} h_{\bar d}(x^*,z^*)\\
-\theta z^* & 0 &\cdots &0 & \lambda 
\end{vmatrix} \\
&= \lambda \begin{vmatrix}
\lambda-\partial_1 g_1(x^*)+\theta z^*& -\partial_2 g_1(x^*) & \cdots &-\partial_{\bar d} g_{1}(x^*)\\
-\partial_1 g_2(x^*) &\lambda-\partial_2 g_2(x^*) & \cdots & -\partial_{\bar d} g_{2}(x^*)\\
& & \vdots & \\
-\partial_1 g_{\bar d}(x^*) &-\partial_2 g_{\bar d}(x^*) & \cdots & \lambda-\partial_{\bar d} g_{\bar d}(x^*)
\end{vmatrix}\\
&-(-1)^{\bar{d}+2}\theta z^* 
\begin{vmatrix}
 -\partial_2 g_1(x^*) & \cdots &-\partial_{\bar d} g_{1}(x^*)&  -\partial_{z} h_{1}(x^*,z^*)\\
\lambda-\partial_2 g_2(x^*) & \cdots & -\partial_{\bar d} g_{2}(x^*)& -\partial_{z} h_{2}(x^*,z^*)\\
& \vdots & &\\
-\partial_2 g_{\bar d}(x^*) & \cdots & \lambda-\partial_{\bar d} g_{\bar d}(x^*)& -\partial_{z} h_{\bar d}(x^*,z^*)
\end{vmatrix}\\
&=\lambda \begin{vmatrix}
\lambda-\partial_1 g_1(x^*)& -\partial_2 g_1(x^*) & \cdots &-\partial_{\bar d} g_{1}(x^*)\\
-\partial_1 g_2(x^*) &\lambda-\partial_2 g_2(x^*) & \cdots & -\partial_{\bar d} g_{2}(x^*)\\
& & \vdots & \\
-\partial_1 g_{\bar d}(x^*) &-\partial_2 g_{\bar d}(x^*) & \cdots & \lambda-\partial_{\bar d} g_{\bar d}(x^*)
\end{vmatrix}\\
&+ \theta z^* \lambda \begin{vmatrix}
1& 0& \cdots &0\\
-\partial_1 g_2(x^*) &\lambda-\partial_2 g_2(x^*) & \cdots & -\partial_{\bar d} g_{2}(x^*)\\
& & \vdots & \\
-\partial_1 g_{\bar d}(x^*) &-\partial_2 g_{\bar d}(x^*) & \cdots & \lambda-\partial_{\bar d} g_{\bar d}(x^*)
\end{vmatrix}\\
&-(-1)^{\bar{d}+2}\theta z^* 
\begin{vmatrix}
 -\partial_2 g_1(x^*) & \cdots &-\partial_{\bar d} g_{1}(x^*)&  -\partial_{z} h_{1}(x^*,z^*)\\
\lambda-\partial_2 g_2(x^*) & \cdots & -\partial_{\bar d} g_{2}(x^*)& -\partial_{z} h_{2}(x^*,z^*)\\
& \vdots & &\\
-\partial_2 g_{\bar d}(x^*) & \cdots & \lambda-\partial_{\bar d} g_{\bar d}(x^*)& -\partial_{z} h_{\bar d}(x^*,z^*).
\end{vmatrix}
\end{aligned}
\end{equation}
Notice that the first term in \eqref{eq:dets} is equal to $\lambda |\lambda I-J(x^*)|$. We denote by $\lambda
G(\lambda), \theta z^* \lambda H_1(\lambda)$ and $(-1)^{\bar{d}+2}\theta z^* H_2(\lambda)$ the first, the second and the third term in \eqref{eq:dets}, respectively. Hence we have
\begin{equation}\label{eq:det2}
|\lambda I- \bar J(x^*,z^*)|=\lambda G(\lambda) + \theta z^* \lambda H_1(\lambda)-(-1)^{\bar{d}+2}\theta z^* H_2(\lambda).
\end{equation}
Now, using the same notations above, we state a theorem related to the stability of $(x^*,z^*)$ of the union system of $(\S,\C,\Re,\K)$ and the basic ACR system \eqref{eq:basic acr for stability}.

\begin{thm}\label{thm:stability}
Suppose the conditions in Lemma \ref{thm:ps exists with basic acr} hold. Suppose further
\begin{enumerate}
\item the associated system for $(\S,\C,\Re,\K)$ admits conservative relations such as \eqref{eq:original conserve} and $u^1_{1}\neq 0$, 
\item all the eigenvalues of $J(x^*)$ have strictly negative real parts, and
\item $H_2(0) > 0 $ if $\bar d$ is odd and $H_2(0) <0$ if $\bar d$ is even.
\end{enumerate}
 Then for sufficiently small $\theta$ and $\mu$ and for sufficiently large $z(0)$, all the eigenvalues of $\bar J(x^*,z^*)$ have also strictly negative real parts. That is the positive steady state $(x^*,z^*)$ is linear stable in the union system of $(\S,\C,\Re,\K)$ and the basic ACR system \eqref{eq:basic acr for stability}.
\end{thm}

\begin{rmk}
Note that the entries $-\partial _z h_i(x^*_1,\dots,x^*_{\bar d},z^*)$ can be calculated by using solely $f_i$ and the conservative relations in the original system. For $i=2,3,\dots,\bar d$,  by the chain rule
\begin{align*}
&\partial_{z} h_{i}(x_1,\dots,x_{\bar d},z)\\
&=-\sum_{j=1}^{k}u^j_1\frac{\partial f_i}{\partial x_{\bar d+j}}\left (x_1,x_2,\dots,x_{\bar d},M_1-\dsum_{j=1}^{\bar d} x_j-z,\dots,M_k-\sum_{j=1}^{\bar d} u^k_{j}x_j(t)-u^k_j z(t) \right ).
\end{align*}
 Similarly for $i=1$,
\begin{align*}
&\partial_{z} h_{1}(x_1,\dots,x_{\bar d},z)\\
&=-\sum_{i=1}^{k}u^j_1\frac{\partial f_i}{\partial x_{\bar d+j}}\left (x_1,x_2,\dots,x_{\bar d},M_1-\dsum_{j=1}^{\bar d} x_j-z,\dots,M_k-\sum_{j=1}^{\bar d} u^k_{j}x_j-u^k_j z \right ) -(\theta x_1-\mu).
\end{align*}
Hence, especially for $(x^*,z^*)$ such that $x^*=\dfrac{\mu}{\theta}$, we have $$-\partial_{z} h_{1}(x^*,z^*)=-\sum_{j=1}^{k}u^j_1\frac{\partial f_i}{\partial x_{\bar d+j}}\left (x^*_1,x^*_2,\dots,x^*_{\bar d},M_1-\dsum_{j=1}^{\bar d} x^*_j-z^*,\dots,M_k-\sum_{j=1}^{\bar d} u^k_{j}x^*_j-u^k_j z^* \right ).$$
\end{rmk}

\begin{proof}
First of all, we scale $\mu=\epsilon^2 \bar \mu, \theta= \epsilon^2 \bar \theta $ and $z(0)=\dfrac{M}{\epsilon}$ for some $R$. Note that $u^1_{1}\neq 0$ by hypothesis 1 in the statement. Then by \eqref{eq:new conv} and \eqref{eq:conv stability of union}, we have
\begin{align*}
z^*&=\bar M_1-x^*_1-\sum_{j=2}^{\bar d} \frac{u^1_j}{u^1_1}x^*_j\\
&=\bar M_1-\frac{\mu}{\theta}-\sum_{j=2}^{\bar d}  \frac{u^1_j}{u^1_1}x^*_j\\
&=z(0)+\sum_{j=1}^{\bar d} \frac{u^1_j}{u^1_1}x_j(0)-\frac{\mu}{\theta}-\sum_{j=2}^{\bar d}  \frac{u^1_j}{u^1_1}x^*_j,
\end{align*}
for each $i$ such that $u_{i1}\neq 0$. Thus for each $\epsilon$, we have 
$$\theta z^*=\epsilon \bar \theta \left ( R + \epsilon \sum_{j=1}^{\bar d} \frac{u_{ij}}{v_1}x_j(0)-\epsilon \frac{\mu}{\theta}-\epsilon\sum_{j=2}^{\bar d}  \frac{u_{ij}}{v_1}x^*_j \right) >0$$ by taking sufficiently large $R=R(\epsilon)$. We denote $c(\epsilon)=\theta z^*$, then $\lim_{\epsilon \to 0}c(\epsilon)=0$.
 
By the hypothesis, all roots of $G(\lambda)$ have strictly negative real parts. Let $\lambda_0,\lambda_1,\dots,\lambda_{\bar d}$ be the roots of $\lambda G(\lambda)$ where $\lambda_0=0$ and $\lambda_i$ are non-zero roots with strictly negative real parts. Let further denote $\lambda_0(\epsilon),\lambda_1(\epsilon),\dots,\lambda_{\bar d}(\epsilon)$ the roots of
$|\lambda I -\bar J(x^*,z^*)|=\lambda G(\lambda)+c(\epsilon)\lambda H_1(\lambda)-(-1)^{\bar d+2} c(\epsilon) H_2(\lambda)$. By the continuity of roots of a polynomial with respect to the coefficients, we have $\lim_{\epsilon \to 0}|\lambda_i(\epsilon)-\lambda_i|=0$ for $i=1,2,\dots,\bar d$. Hence with small enough $\epsilon$, we could make the real parts of $\lambda_i(\epsilon)$ is still negative for each $i=1,2\dots,\bar d$.

We turn to show that $\lambda_0(\epsilon)$ has also a strictly negative real part. Note that $|\lambda I -\bar J(x^*,z^*)|\Big |_{\lambda=0}=(\lambda-\lambda_0(\epsilon))(\lambda-\lambda_1(\epsilon))\cdots (\lambda-\lambda_{\bar d}(\epsilon))$. Hence
\begin{equation}\label{eq:product of roots}
(-1)^{\bar d+1}\lambda_0(\epsilon)\lambda_1(\epsilon)\cdots\lambda_{\bar d}(\epsilon)=-(-1)^{\bar d+2}c(\epsilon)H_2(0).
\end{equation}
Suppose $\lambda_0(\epsilon)$ is a complex number with non-zero imaginary part. Then it must be a conjugate of $\lambda_i(\epsilon)$ for some $i$. Since we choose $\epsilon$ small enough so that the real part of each $\lambda_i(\epsilon)$ is strictly negative, $\lambda_0(\epsilon)$ has a negative real part. Now we suppose that $\lambda_0(\epsilon)$ is a real number. In this case, because of the negative real parts of $\lambda_i(\epsilon)$, the product $\prod_{i=1}^{\bar d}\lambda_i(\epsilon)$ is negative if $\bar d$ is odd and is positive otherwise. Thus by the hypothesis 2 and \eqref{eq:product of roots}, we have $\lambda_0(\epsilon) <0$. Thus the result follows.
\end{proof}
\textcolor{black}{It is notable that the condition of sufficient amount of initial Z in Theorem \ref{thm:stability} is known to be often satisfied in practical applications. }

\section{Control of networks with no positive steady states}\label{sec:control general 2d}
In this section, we introduce an ACR system that controls both a target species and other species in a given network system. By using this type of expanded ACR system, we show that a 2-dimensional reaction system, which admits no positive steady states, can be controlled as we showed in Section \ref{sec:additional species} of the main text. For a two dimensional system with species $A$ and $B$, let $A$ be the target species we desire to control. Then we define a mass action ACR system \eqref{eq:expanded acr system} that controls the other species $B$ as well as the target species $A$,
\begin{equation}\label{eq:expanded acr system}
A+Z\xrightarrow{\theta} 2Z, \quad Z\xrightarrow{\mu} A, \quad 
B+Z \xrightleftharpoons[\alpha_2]{\alpha_1} Z.
\end{equation}

\begin{lem}\label{lem:control 2d unstable}
Let $(\S,\C,\Re,\K)$ be a network system such that $\S=\{A,B\}$. Let $x(t)=(a(t),b(t))$ be the associated deterministic system such that 
\begin{align*}
\frac{d}{dt}a(t)=f_1(a(t),b(t)), \quad \text{and} \quad \frac{d}{dt}b(t)=f_2(a(t),b(t)).
\end{align*}
Suppose that there is no positive steady state for $(\S,\C,\Re,\K)$. Suppose further that for any positive constant $c$, there exists $d>0$ such that $f_1(c,d)=0$. Then the union system of $(\S,\C,\Re,\K)$ and \eqref{eq:expanded acr system}  admits a positive steady state $(a^*,b^*,z^*)$ such that $a^*=\frac{\mu}{\theta}$ for any $\mu$ and $\theta$ provided $\alpha_1b^*>\alpha_2$ when $f_2(\frac{\mu}{\theta},b^*)>0$ and $\alpha_1b^*<\alpha_2$ when $f_2(\frac{\mu}{\theta},b^*)<0$.
\end{lem}

\begin{proof}
First, the reactions $B+Z \rightleftharpoons Z$ do not change the concentration of $Z$. Hence $A$ is still an ACR species as we showed around \eqref{eq:basic acr}. Let $\bar x=(\bar a,\bar b,z)$ be the solution to the deterministic system associated with the union system of $(\S,\C,\Re,\K)$ and \eqref{eq:expanded acr system}. Then we have
\begin{align*}
&\frac{d}{dt}\bar a(t)=\bar f_1(\bar a,\bar b,z)=f_1(\bar a,\bar b)-z(\theta \bar a-\mu),  \\
&\frac{d}{dt}\bar b(t)=\bar f_2(\bar a,\bar b,z)=f_2(\bar a,\bar b)-z(\alpha_1 \bar b-\alpha_2), \text{ and}\\
&\frac{d}{dt}z(t)=z(\theta \bar a-\mu).
\end{align*}

By the hypothesis, there exists a positive constant $x^*_2$ such that $\bar f_1(\frac{\mu}{\theta},x^*_2,z)= f_1(\frac{\mu}{\theta},x^*_2)=0$ for any $z$. 
Then if we choose $z^*=\dfrac{f_2(\frac{\mu}{\theta},x^*_2)}{\alpha_1  x^*_2-\alpha_2}$, we have $\bar f_2(\frac{\mu}{\theta},x^*_2,z^*)=0$. Since $(\S,\C,\Re,\K)$ does not admits a positive steady state, $f_2(\frac{\mu}{\theta},x^*_2)\neq 0$. Thus $z^*\neq 0$.

It follows that $z^*>0$ because we assumed that $\alpha_1 x^*_2>\alpha_2$ if $f_2(\frac{\mu}{\theta},x^*_2)>0$ and $\alpha_1 x^*_2<\alpha_2$ if $f_2(\frac{\mu}{\theta},x^*_2)<0$. Therefore $(\frac{\mu}{\theta},x^*_2,z^*)$ is a positive steady state of the union system.
\end{proof}

\begin{rmk}
In the proof above, it was shown that $b^*$ is independent of $\alpha_1$ and $\alpha_2$. Hence we can evaluate the value of $f_2(\frac{\mu}{\theta},b^*)$ and then we set the parameters $\alpha_1$ and $\alpha_2$ in the controller \eqref{eq:expanded acr system} according to the sign of $f_2(\frac{\mu}{\theta},b^*)$.
\end{rmk}

\begin{rmk}
\textcolor{black}{For arbitrary $\theta>0$ and $\mu>0$, if there exists a positive steady state for the controlled system $\bar a(t)$, $\bar b(t)$ and $z(t)$, then it must hold $\bar f_1(\frac{\mu}{\theta},b^*,z^*)=f_1(\frac{\mu}{\theta},b^*)=0$ for some $b^*$. Hence for the controllablility of the basic ACR system, it is a necessary condition of a 2-dimensional system that for any $c>0$ there exists $d>0$ such that $f_1(c,d)=0$.}
\end{rmk}

\begin{rmk}
\textcolor{black}{Stability analysis of the union system with an extended ACR controller is a open problem.}
\end{rmk}

\textcolor{black}{For a higher-dimensional system, we can use a similar idea to build an extended ACR circuit with multiple control species. }

\begin{example}
\textcolor{black}{
Consider the following $3$-dimensional mass-action system admitting no positive steady states, where protein $B$ and protein $C$ are dimerized to a protein $A$.
\begin{align*}
B+C \xrightarrow{1} A, \quad B\xleftarrow{1}&\ 0 \xrightarrow{2} C\\
&\uparrow \text{\footnotesize{$1$}}\\
&A
\end{align*}
Then, we can consider an extended ACR system with two control species $Z_1$ and $Z_2$ such that
\begin{align*}
&A+Z_1\xrightarrow{\theta} 2Z_1, \quad Z_1\xrightarrow{\mu} A,\\
&A+Z_2\xrightarrow{\theta} 2Z_2, \quad Z_2\xrightarrow{\mu} A,\\
&B+Z_1 \xrightleftharpoons[\alpha_2]{\alpha_1} Z_1, \quad C+Z \xrightleftharpoons[\alpha_4]{\alpha_3} Z_1.
\end{align*}
Let $\dfrac{\mu}{\theta}=10$ be the desired set point for $A$ in the union of the given $3$-dimensional system and the extended ACR system. 
Then species $A$,$B$ and $C$ are stabilized at $(10,b^*,c^*)$ such that $b^*c^*=10$, as long as the parameters $\alpha_1,\alpha_2,\alpha_3$ and $\alpha_4$ satisfy 
\begin{align*}
\alpha_2 > b^*\alpha_1 \quad \text{and} \quad \alpha_4 > b^*\alpha_3.
\end{align*}}
\begin{figure}
\centering{
\includegraphics{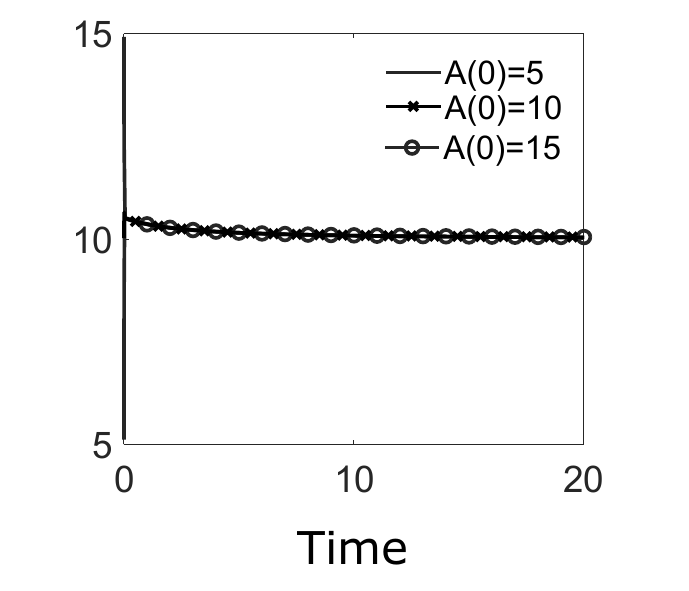}}
\caption{Time evolution of the concentration of $A$ in the union system.}
\end{figure}
\end{example}

\section{Applications of the deterministic results}
\subsection{ERK system}\label{subsec:ERKsystem}
In this section, we analyze the stability of the controlled ERK system shown in Figure~\ref{fig:ERK}~a.
We show that both the conditions in Lemma \ref{thm:ps exists with basic acr} and conditions in Theorem \ref{thm:stability} hold for a positive steady state $(x^*,z^*)$ in the controlled ERK system. Let $(\S,\C,\Re,\K)$ be the deterministic system associated with the ERK network in Figure \ref{fig:ERK} a using the parameters given in the main text. Let also $(\bar \S,\bar \C, \bar \Re, \bar \K)$ be the union system of $(\S,\C,\Re,\K)$ and the ACR controller in \ref{fig:ERK} a with $\theta=1$ and $\mu=2$. We use a Matlab simulation to obtain a positive steady state, as well as the relevant Jacobian, eigenvalues and determinant.

Let $x(t)=(x_1(t),x_2(t),\dots,x_{12}(t))$ and $\bar x(t)=(\bar x_1(t),\bar x_2(t),\dots,\bar x_{12}(t),z(t))$ represent the concentrations of species in the systems $(\S,\C,\Re,\K)$ and $(\bar \S,\bar \C, \bar \Re, \bar \K)$, respectively. We arrange $x_1,x_2,\dots,x_{12}$ so that they represent $F, E, S_{00}, S_{01}, S_{10}, ES_{00}, ES_{01}, FS_{01}, FS_{10}, S_{11}, ES_{10}$ and $FS_{11}$, respectively. We also let $\bar x_i$ represent the same concentration as $x_i$, and we let $z$ represent the concentration of $Z$. 

To show the condition in Lemma \ref{thm:ps exists with basic acr}, we show that the system $(\S,\C,\Re,\K)$ admits a positive steady state at $x^*$ such that $x^*_1=\frac{\mu}{\theta}=2$. We verify the existence of the positive steady state using the simulation shown in the figure below for the ERK system with $F_{\text{tot}}=31, E_{\text{tot}}=37, S_{\text{tot}}=100$.

\begin{figure}[H]
\centering
\includegraphics{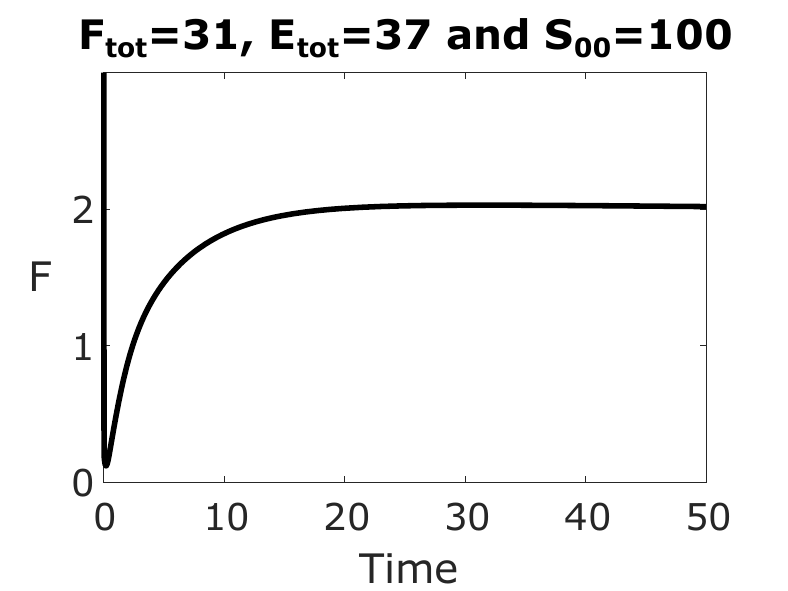}
\end{figure}
The positive steady state $x^*=(2.0, 2.1,7.8, 0.9, 11.9, 16.7, 7.6, 1.5, 15.2,15.3,10.9,12.3)$. We rounded off the values to one decimal place.

We also notice there are three linear independent conservative relations in $(\S,\C,\Re,\K)$,
\begin{equation}\label{eq:conv erk}
\begin{split}
x_{12}(t)&=F_{\text{tot}}-x_1(t)-x_8(t)-x_9(t),\\
x_{11}(t)&=E_{\text{tot}}-x_2(t)-x_6(t)-x_7(t),\\
x_{10}(t)&=S_{\text{tot}}-\sum_{i=3}^{12}x_i(t).
\end{split}
\end{equation}
 The target species $F$ is involved in the first conservative relation in \eqref{eq:conv erk}. Hence condition 1 in Theorem \ref{thm:stability} holds.

For the second condition in Theorem \ref{thm:stability}, we have the following Jacobian $J(x^*)$ of the reduced system obtained by using the conservation laws \eqref{eq:conv erk} for $(\S,\C,\Re,\K)$ as \eqref{eq:reduce conserve}. 
\begin{equation}\label{eq:jac for erk}
\textcolor{black}{J(x^*)=\begin{bmatrix}
-25.3  & -3.0 &   3.0&   -3.0&    1.5&         0 &        0&    3.0&    2.0\\
         0&  -90.8&   -1.0&   1&   -5.0&   -4.0&         0&         0\\
    0 & -56.2&   -1.0&         0&         0&    2.0&         0&    1.0&    1.0\\
       -4.4&   -3.3&         0   -7.0&         0&         0&    1.0&    4.0&         0\\
   -8.1&  -29.2&         0&         0&   -2.5&   -5.0&   -5.0&         0&    3.0\\
         0&   56.2&    1.0&         0&         0&   -3.0&         0&         0&         0\\
         0 &   3.3&         0&    1.0&         0&    1.0&   -3.0&         0&         0\\
    4.4&         0&         0&    6.0&         0&         0&         0&   -5.0&         0\\
    5.1&         0&         0&         0&    1.5&         0&         0&   -3.0&   -7.0
\end{bmatrix}}
\end{equation}
The eigenvalues of $J(x^*)$ are
 $-0.11,  -0.8,   -1.4,    -4.1,  -6.0, -7.5,  -9.3,  -27.0,  -88.2$.

For the third condition in Theorem \ref{thm:stability}, we note that $d=12, k=3$ and hence $\bar d=9$. Thus, if $H_2(0)>0$, then the condition holds. Note that we have three conservation relations for the system $(\bar \S,\bar \C,\bar \Re,\bar \K)$ 
\begin{align*}
x_{12}(t)&=F_{\text{tot}}-x_1(t)-x_8(t)-x_9(t)-z(t),\\
x_{11}(t)&=E_{\text{tot}}-x_2(t)-x_6(t)-x_7(t),\\
x_{10}(t)&=S_{\text{tot}}-\sum_{i=3}^{12}x_i(t).
\end{align*}
Hence using the same notation in Section \ref{sec:stability}, we have
\begin{align*}
\partial_{z} h_{i}(x^*_1,\dots,x^*_{\bar d})=\begin{cases}
-2 \quad &\text{if $i=1$},\\
-3 &\text{if $i=9$, and}\\
0 &\text{otherwise}.
\end{cases}
\end{align*}
Combining this with \eqref{eq:jac for erk}, we obtain the matrix 
\begin{align*}
\begin{pmatrix}
-\partial_2 g_1(x^*_1,\dots,x^*_{\bar d}) & \cdots &-\partial_{\bar d} g_{1}(x^*_1,\dots,x^*_{\bar d})&  -\partial_{z} h_{1}(x^*_1,\dots,x^*_{\bar d})\\
\lambda-\partial_2 g_2(x^*_1,\dots,x^*_{\bar d}) & \cdots & -\partial_{\bar d} g_{2}(x^*_1,\dots,x^*_{\bar d})& -\partial_{z} h_{1}(x^*_1,\dots,x^*_{\bar d}))\\
 & \vdots & \\
-\partial_2 h_{\bar d}(x^*_1,\dots,x^*_{\bar d}) & \cdots & \lambda-\partial_{\bar d} g_{\bar d}(x^*_1,\dots,x^*_{\bar d})& -\partial_{z} h_{\bar d}(x^*_1,\dots,x^*_{\bar d}))
\end{pmatrix} 
\end{align*}
shown in \eqref{eq:dets}. By plugging in $\lambda=0$ into the matrix and computing its determinant, we have $H_2(0)=3.3\times 10^6$. Since $\bar d=9$, the third condition hold in Lemma \ref{thm:ps exists with basic acr}.

Consequently we show that all the conditions in Lemma \ref{thm:ps exists with basic acr} and show that Theorem hold \ref{thm:stability}, hence the linear stability of $(\S,\C,\Re,\K)$ follows.

\subsection{A 2-dimensional system admitting no positive steady states}\label{subsec:Nonconserv system}
In this section, we use the expanded ACR controller \eqref{eq:expanded acr system} to control the 2-dimensional system 
\begin{equation}\label{eq:system without ps}
\begin{split}
&\hspace{0.6cm} A	\\[-1ex]
\text{\scriptsize{$3$}} \hspace{-0.2cm}&\nearrow \\[-0.6ex]
A+B \xrightarrow{1} \emptyset & \\[-1ex]
\text{\scriptsize{$5$}} \hspace{-0.2cm}&\searrow  \\[-1ex]
&\hspace{0.6cm} B,
\end{split}
\end{equation}
 which is introduced in Section \ref{subsec:stochastic control} of the main text.
 The mass-action deterministic system associated with this network is
 \begin{align*}
\frac{d}{dt}a(t)&=f_1(a,b)=-a(t)b(t)+3,\\
\frac{d}{dt}b(t)&=f_2(a,b)=-a(t)b(t)+5.
\end{align*}
Note that this system does not admit a positive steady state, as there does not exists $(a^*,b^*)\in \R^2_{>0}$ such that $3-a^*b^*=0$ and $5-a^*b^*=0$. In particular, $\dlim_{t\to \infty}(b(t)-a(t))=\infty$ since $\dfrac{d}{dt}(b(t)-a(t))=2$. 

Furthermore the union system of \eqref{eq:system without ps} and the basic ACR system $Z+A\xrightarrow{1} 2Z$ and $Z\xrightarrow{5} A$ also does not admit a positive steady state. The associated mass-action system is
\begin{align*}
\frac{d}{dt}a(t)&=-a(t)b(t)+3-z(t)(a(t)-5),\\
\frac{d}{dt}b(t)&=-a(t)b(t)+5,\\
\frac{d}{dt}z(t)&=z(t)(a(t)-5).\\
\end{align*}
Suppose there exists a positive steady state $(a^*,b^*,z^*)$. By the last equation, $a^*= 5$. Plugging $a^*= 5$ into $a(t)$ in the first equation, we have $b^*=\frac{3}{5}$. However, at this state, $b$ is not stabilized as $-a^*b^*+5=2$, hence it contradicts to the assumption that $(a^*,b^*,z^*)$ is a positive steady state.

Now we consider the union system of \eqref{eq:system without ps} and the expanded ACR controller \eqref{eq:expanded acr system} introduced in Section \ref{sec:control general 2d}.
\begin{align*} 
Z+A\xrightarrow{\theta} 2Z, \quad Z\xrightarrow{\mu} A, \quad Z+B \xrightleftharpoons[\alpha_2]{\alpha_1} Z,
\end{align*} with general positive parameters $\kappa_1, \kappa_2, \kappa_3, \theta, \mu, \alpha_1$ and $\alpha_2$. We use Lemma \ref{lem:control 2d unstable} to show this union system admits a positive steady state under a mile condition.
The associated mass-action system is
\begin{align*}
\frac{d}{dt}a(t)&=-a(t)b(t)+3-z(t)(\theta a(t)-\mu),\\
\frac{d}{dt}b(t)&=-a(t)b(t)+5-z(t)(\alpha_1 b(t)-\alpha_2),\\
\frac{d}{dt}z(t)&=z(t)(\theta a(t)-\mu).\\
\end{align*}
It can be easily shown that for any positive constant $c$, there exists $d=\dfrac{3}{d}$ such that $f_1(c,d)=0$. Hence for $a^*=\dfrac{\mu}{\theta}$ with arbitrary $\mu>0$ and $\theta>0$, there exists $b^*=\dfrac{3\theta}{\mu}$ such that $f_1(a^*,b^*)=0$. Since $f_2(a^*,b^*)=2>0$, we tune the parameters $\alpha_1>0$ and $\alpha_2>0$ in \eqref{eq:expanded acr system} as they satisfy $\alpha_1 b^* > \alpha_2$. 

Hence by Lemma \ref{lem:control 2d unstable}, 
\begin{align*}
(a^*,b^*,z^*)=\left (\frac{\mu}{\theta},\frac{\kappa_2 \theta}{\kappa_1 \mu},\frac{\kappa_1 \mu (\kappa_3-\kappa_2)}{\alpha_1 \kappa_2 \theta-\alpha_2 \kappa_1 \mu} \right )
\end{align*} is a positive steady state.
\section{Stochastic ACR control}\label{sec:control stochastic}
To control a stochastic system via an ACR controller, we rely on an approximation under multiscaling model reduction as described in Section \ref{subsec:stochastic control} and \ref{subsec:faststo} of the main text. 
In this section we introduce the formal procedures for generating a reduced model, and we introduce related theorems.
\subsection{Network reduction}
To formally define a reduced model of a given stochastic system, a notion of network projection needs to be introduced. The reduced system shown in Figure \ref{fig:RL schematic} b and the hybrid system shown in \ref{fig:faststochastic} d are obtained through network projection. For example,
consider the reaction network
\begin{equation}\label{eq:orginal network}
A \xrightleftharpoons[\kappa_1]{\kappa_2} B.
\end{equation}
Suppose that for some parameters, species $B$ is rarely produced or removed until time $t=1$. In this case,  we approximate the distribution of $A$ with the stochastic system associated with
\begin{equation}\label{eq:projection}
A \xrightleftharpoons[\kappa_1]{\kappa_2 B(0)} 0.
\end{equation}
Note that this reduced network is obtained by freezing the copy number of $B$ at $B(0)$. We call network \eqref{eq:projection} the projection of \eqref{eq:orginal network} by freezing species $B$ at $B(0)$. As this example shows, network projection can be used to describe an asymptotic behavior of a subset of species.

We define a projection function for complexes and reactions in $(\S,\C,\Re)$ with $\S=\S_L\cup \S_H$ where $\S_L=\{S_1,S_2,\dots,S_d\}$ and $\S_H=\{S_{d+1},S_{d+2},\dots,S_{d+r}\}$. In the later section, $\S_L$ and $\S_H$ would represent collections of species with low and high copy numbers, respectively. Let $q_L : \Z^{d+r}\to \Z^d$ and $q_H : \Z^{d+r}\to \Z^r$ be projection function such that for each $v=(v_1,\dots,v_d,v_{d+1},\dots,v_{d+r})^T\in \Z^{d+r}$,
\begin{align*}
q_L(v)=(v_1,v_2,\dots,v_d)^T\in \Z^d \quad \text{and} \quad q_H(v)=(v_{d+1},v_{d+2},\dots,v_{d+r})^T \in \Z^r.
\end{align*}

We use the projection function $q_L$ and $q_H$ for complexes and reaction. For example, for a network $A+B\to B$ with complexes $A$ and $B$, we let $\S_L=\{A\}$. Then the complex $A$, $B$ and the reaction $A+B\to B$ are represented with two dimensional vectors $(1,1)^T, (0,1)^T$ and $(-1,0)$, respectively. Then the projection of $A$, $B$ and the reaction $A+B\to B$ are
$q_L((1,1)^T)=1$, $q_L((0,1)^T)=0$ and $q_L((-1,0)^T)=-1$ which are associated with complexes $A$, $0$ and reaction $A\to 0$, respectively. Hence by abusing notation, $q_L(A)=A, q_L(B)=0$ and $q_L(A+B\to B)=A\to 0$. Generally, we denote $q_L(y)$ the complex obtained by projection of the complex vector associated with a complex $y$. In this way, the projected network $(\S_L,\C_L,\Re_L)$ of the original reaction network $(\S,\C,\Re)$ by $q_L$ is defined to be 
\begin{equation}\label{eq:project network}
\begin{split}
&\S_L=\{S_1,\dots,S_d\}, \ \C_L=\{q_L(y):y\in \C\}, \ \text{and} \\ 
 &\Re_L=\{ q_L(y)\to q_L(y') : y\to y' \in \Re \ \text{ such that } q_L(y')-q_L(y)\neq \overrightarrow{0}\}.
 \end{split}
\end{equation}

The rate constants of the projected network are defined with a given rate constants $\K$ of a given network $(\S,\C,\Re)$. We inherit $\K$ to the projected network by incorporating the terms coming from freezing species $\S_H$ at their initial count. For example, the rate constant of reaction $A\to 0$ in \eqref{eq:projection} is $\kappa_2 B(0)$, where $\kappa_2$ is inherited from reaction $A+B\to B$.

Hence for the set of reaction intensities $\K$ in a given network $(\S,\C,\Re,\K)$, the set of reaction intensities for a projected network is
\begin{equation}\label{eq:rate constants for reduced}
\K_L=\left \{\bar \lambda_u(x)=\sum_{\substack{y_k\to y'_k \in \Re \\q_L(y_k)=\bar y_u, q_L(y'_k)=\bar y'_u}}q_H(X(0))^{q_H(y)} \lambda_k(x)  : \bar y_u \to \bar y'_u \in \Re_L \right \},
\end{equation}
 where $u^v=\prod u_i^{v_i}$ for the same dimensional non-negative vectors $u$ and $v$ and we use here the convention $0^0=1$. The summation in \eqref{eq:rate constants for reduced} arises when multiple reactions in $(\S,\C,\Re)$ are projected into the same reaction in $\Re_L$.


\subsection{Stationary distribution approximation under multiscaling model reduction}\label{subsec:model reduction}
The main idea of the stationary distribution approximation shown in Section \ref{subsec:stochastic control} in the main text is the multiscaling model reduction. In this section, we introduce the multiscaling for a stochastic reaction network system with reaction propensities of constant order. 

Throughout this section we use the following notations. Let $(\S,\C,\Re,\K)$ be a network system with $\S=\{S_1,S_2,\dots,S_{d+r}\}$ and let $N$ be a scaling parameter. Let $X^N=(X^N_1,X^N_2,\dots,X^N_{d+r})$ be the associated scaled stochastic process with transition probabilities \eqref{eq:prob} such that each $X^N_i$ represents the counts of species $S_i$. For a given initial condition $X^N(0)$, let 
\begin{equation}\label{eq:S setting}
\S_L=\{S_i \in \S : X^N_i(0) =c_i\}\quad \text{and} \quad \S_H=\{S_i \in \S: X^N_i(0)=c_i N\},
\end{equation} where $c_i$'s are positive constants.

Then for a given collection of rate constants $\K$, we scale the system with the following reaction intensities.
\begin{equation}\label{eq:scaled rates}
\K^N=\left \{\lambda^N_{k} = \frac{\kappa_{y\to y'}}{N^{\Vert q_H(y) \Vert_1}}\lambda_k : \lambda_k \in \K \right \},
\end{equation}
where $\Vert \cdot \Vert_1$ is the 1-norm. 

For example, let $A+B\xrightarrow 0$ be a reaction network system with $\S_L=\{A\}$ and $\S_H=\{B\}$. Let $x=(x_A,x_B)$ be a state of the associated stochastic process $X^N$ with a scaling parameter $N$. Then the for a given reaction intensity $\lambda(x)=\kappa x_A x_B$ of the reaction $A+B\to 0$, we define a scaled reaction intensity $\lambda^N(x)=\frac{\kappa}{N^{\Vert q_H(A+B)\Vert_1}}\lambda(x)=\frac{\kappa}{N}x_A x_B$. Then the scaled stochastic process $X^N$ has the transition probability
\begin{align*}
P(X^N(t+\Delta t)=x+(-1,-1)^T \ | \ X^N(t)=x)=\lambda^N(x) \Delta t + h(\Delta),
\end{align*} where $h$ is defined as \eqref{eq:prob}.

Having presented the above example, we now describe in more detail the procedure for a formal multiscaling model reduction. Note that in the example above, as long as $X_B(t)$ is of order $N$, the propensity is of order 1. Under this circumstance, we intuitively expect that $X_B$ would not be substantially change because of the relatively low reaction propensity. Using this background, we can approximate the distribution of a stochastic system through multiscaling model reduction. 
The proof of the following theorem is provided in the separate paper \cite{EK2019}.

\begin{thm}\label{thm:main1} 
Let $X^N=(X^N_1,X^N_2,\dots,X^N_{d+r})$ be the stochastic processes associated with $(\S,\C,\Re,\K^N)$ where $\K^N$ is as \eqref{eq:scaled rates}. For an initial condition $X(0)$, suppose $\S=\S_L\cup \S_H$ with $\S_L$ and $\S_H$ as in \eqref{eq:S setting}. Let $X$ be the associated stochastic process for the projected network system $(\S_L, \C_L, \Re_L, \K^N_L)$. Let further that  $p_L^N(\cdot,t)$ and $p(\cdot,t)$ be the probability distributions for $q_L(X^N)$ and $X$ at time $t$, respectively. Then for any $A \subset \Z^{d}_{\ge 0}$, we have
\begin{equation}\label{eq:conv}
|p^N(A,t)- p(A,t)| = O(N^\nu) \quad \text{for some } \nu \in (0,1).
\end{equation}
\end{thm}
\begin{rmk}
In Theorem \ref{thm:main1}, if $X$ admits a stationary distribution $\pi$, then 
\begin{equation}\label{eq:conv2}
|p^N(A,t)- \pi(A)| \le |p^N(A,t)- p(A,t)| +|p(A,t)-\pi(A)|.
\end{equation}
Hence, for fixed $t$, if $|p(A,t)-\pi(A)|$ is sufficiently small, then $p_L^N(A,t)\approx \pi(A)$ for large $N$. 
\end{rmk}
For a special case, $p_L^N(A,t)$ could be explicitly estimated with a Poisson distribution $\pi$
The following corollary follows by Theorem \ref{thm:prod} and Theorem \ref{thm:main1}.

\begin{cor}\label{cor:poissonapprox}
Under the same conditions in Theorem \ref{thm:main1}, suppose that $(\S_L, \C_L, \Re_L)$ has zero deficiency and is weakly reversible reaction network. Then for a positive steady state $c\in \Re^{|\S_L|}_{>0}$ of the deterministic counter part \eqref{eq:deterministic system} of $(\S_L,\C_L,\Re_L,\K_L)$, we have
\begin{equation}\label{eq:poisson approx}
\lim_{t\to \infty}\lim_{N\to \infty}p_L^N(x,t)= M\prod_{i=1}^{|\S_L|}\frac{c_i^{x_i}} {x_i!} \mathbbm{1}_{ \{x\in \mathbb{S}\} }
\end{equation}
 where $\mathbb{S}$ is the state space, and $M$ is the normalizing constant. 
\end{cor}

\subsection{Mean of the projected systems}\label{subsec:mean of reduced}
In this section we consider the mean of the target species in a reduced system. 
As described in Theorem \eqref{thm:prod}, when the reduced network $(\S_L,\C_L,\Re_L)$ has zero deficiency and is weakly reversible, the long-term behavior of the associated stochastic system follows a product form of Poissons (or constrained Poissons). Moreover the mean of the system is determined by the positive steady state of the deterministic counterpart as described in Theorem \eqref{thm:prod}. Hence, we control the mean of the target species approximately by using the positive steady state value of the reduced system. In fact, for a certain reduced network obtained from an ACR system, the steady state value of an ACR species is preserved in the reduced network.

Let $\pi$ be a product form of Poisson distribution (or constrained Poissons) such as $\pi(x)=M \prod_{i=1}^{d}\frac{c_i^{x_i}}{x_i!} \mathbbm{1}_{ \{x\in \mathbb{S}\} }$ defined on a state space $\mathbb{S}$ with some normalizing constant $M$ and some $c\in \R^d_{\ge 0}$. Let $\pi_1$ be the marginal distribution of the $x_1$-coordinate. If the support of $\pi_1$ is whole $\Z_{\ge 0}$ then $\pi_1$ is the Poisson distribution with rate $c_1$ because
\begin{align*}
\pi_1(x_1)=\sum_{x_2 \in Z_{\ge 0}}\cdots \sum_{x_d \in Z_{\ge 0}} \pi(x) = \frac{1}{\sum_{x_i=0}^\infty \frac{c_1^{x_1}}{x_1!}} \frac{c_1^{x_1}}{x_1 !}=e^{-c_1}\frac{c_1^{x_1}}{x_1 !}
\end{align*}
Thus the mean of the marginal distribution $\pi_1$ is equal to $c_1$. Note that the reactions in the basic ACR system \eqref{eq:basic acr} are always projected to $
0 \xrightleftharpoons[\theta]{\mu} X_1$, the state space of $X_1$ is always equal to $Z_{\ge 0}$ in the projected network system. Hence the mean of the target species in the controlled system is close to $c_1$, which is the positive steady state of the deterministic counter part.

Similarly to the stability analysis carried out in Section \ref{sec:stability}, we assume a given network system $(\S,\C,\Re,\K)$ is ACR and admits conservation relations. With the same notation we used in Section \ref{sec:stability}, let
\begin{equation}\label{eq:conv for mean}
u^i\cdot x(0)=u^i\cdot x(t) = M_i \quad \text{for all $t$}
\end{equation} for some positive constants $M_i$'s and for some vectors $u^i$'s in $\R^d_{> 0}$. 
Hence for each positive steady state $x^*$ of the system, if the $i$-th entry of $x^*$ corresponds to a non-ACR species, then $x^*_i$ depends on the total concentrations $M_i$. In this case, we denote 
$x^*(M)$ a positive steady state on the stoichiometry class $S_{x(0)}$, where $M=(M_1,M_2,\dots,M_k)$ such that $ u^i\cdot x(0)=M_i$. We also denote the entry of $x^*(M)$ by $x^*_i=x^*_i(M)$.
With these notations, the following theorem states that if a reduced network is obtained by freezing some non-ACR species of an ACR system, then the steady state value of an ACR species is preserved in the reduced system.

\begin{thm}\label{thm:steady state after freezing}
Let $(\S,\C,\Re,\K)$ be a mass-action ACR system with $x'(t)=f(x(t)).$
Suppose $(\S,\C,\Re,\K)$ admits the conservation laws \eqref{eq:conv for mean}. We split $\S=\{S_1,\dots,S_d,S_{d+1},\dots,S_{d+r}\}$ into two disjoint subsets $\S_L=\{S_1,\dots,S_d\}$ and $\S_H=\{S_{d+1},\dots,S_{d+r}\}$, where none of the species in $S_H$ is an ACR species. 
Let $\{x^*(M): M=(M_1,M_2,\dots,M_k)\in R^k_{>0}\}$ be the family of positive steady states of the system  $(\S,\C,\Re,\K)$. 
For the dynamics $x(t)$ of $(\S,\C,\Re,\K)$, suppose there exist a $\widetilde M=(\widetilde M_1,\dots,\widetilde M_k)$ such that
\begin{equation}\label{eq:condition of mean}
x_{d+i}(0) = x^*_{d+i}(\widetilde M_1,\dots,\widetilde M_k) \quad \text{for $i=1,2,\dots,r$}.
\end{equation}
Then the projected system $(\S_L,\C_L,\Re_L,\K_L)$ with the initial condition $q_L(x(0))$ has a positive steady state value at $\bar x^* = q_L(x^*(\widetilde M))$.
In particular, if $S_i\in S_L$ is an ACR species with the ACR value $x^*_i$, then $\bar x^*_i = x^*_i$.
\end{thm}
\begin{proof}
Let $\bar x'(t)=\bar f (\bar x(t))$ be the deterministic system of $(\S_L,\C_L,\Re_L,\K_L)$.
Note that $(\S_L,\C_L,\Re_L)$ is obtained by freezing the species in $\S_H$ from  $(\S,\C,\Re)$. Recall by the definition of $\K_L$ \eqref{eq:scaled rates} that each element $\K_L$ is a summation of some rate constants in $\K$ multiplied by the initial values of species in $\S_H$. Hence for any positive $x_i$'s, we have
$$\bar f_i(x_1,\dots,x_d)=f_i(x_1,\dots,x_d,x_{d+1}(0),\dots,x_{d+r}(0)),\quad  i=1,2,\dots,d.$$ 
For the $\widetilde M$, we have 
\begin{align*}
0=f(x^*)&=f(x^*_1(\widetilde M),\dots,x^*_d,x^*_{d+1}(\widetilde M),\dots,x^*_{d+r}(\widetilde M))\\
&=f(x^*_1(\widetilde M),\dots,x^*_d(\widetilde M),x_{d+1}(0),\dots,x_{d+r}(0))\\
&=\bar f(x^*_1(\widetilde M),\dots,x^*_d(\widetilde M)).
\end{align*}
Hence $\bar x^*=(x^*_1(\widetilde M),\dots,x^*_d(\widetilde M))=q_L(x^*(\widetilde M))$ is a positive steady state of $(\S_L,\C_L,\Re_L,\K_L)$. 

Suppose $S_i\in S_L$ is an ACR species with the ACR value $x^*_i$. Since the ACR value is independent of any conservative quantities $M_1, M_2, \dots,M_k$ by definition, we have $x^*_i(\widetilde M)=x^*_i$. Consequently, $\bar x^*_i=x^*_i$.
\end{proof}

\begin{example} \textcolor{black}{Consider a network $(\S,\C,\Re,\K)$ }
\tikzset{every node/.style={auto}}
 \tikzset{every state/.style={circle, rounded corners, minimum size=1pt, draw=none, font=\normalsize}}
 \tikzset{bend angle=8}
 \begin{equation*}
   \begin{tikzpicture}[baseline={(current bounding box.center)}, scale=0.8, rounded corners]
   \node[text=black ] (1) at (0,4)  {\textbf{$C+A$}};
    \node[text=black ] (2) at (3,4)  {\textbf{$D+A$}};

   \node[ text=black] (4) at (-3,6)  {\textbf{$A$}};
   \node[ text=black] (5) at (0,6)  {\textbf{$B$}};

   \node[ text=black] (6) at (2,6)  {\textbf{$Z+A$}};
   \node[ text=black] (7) at (4.5,6)  {\textbf{$2Z$}};
   \node[ text=black] (8) at (6.5,6)  {\textbf{$Z$}};
     \node[ text=black] (9) at (8.5,6)  {\textbf{$A$}};
   
   \path[shorten >=2pt,->,shorten <=2pt]
   (1) edge[bend left, line width=1pt] node {$1$} (2)

    (2) edge[bend left, line width=1pt] node {$1$} (1)
    
    (4) edge[bend left, line width=1pt] node {$1$} (5)
    (5) edge[bend left, line width=1pt] node {$1$} (4)
    (6) edge[line width=1pt] node {$\theta$} (7)
    (8) edge[ line width=1pt] node {$\mu$} (9)
       ;
    
  \end{tikzpicture}
 \end{equation*} 
\textcolor{black}{ There are two conservation relations $a(0)+b(0)+z(0)=a(t)+b(t)+z(t)=M_1$ and $c(0)+d(0)=c(t)+d(t)=M_2$. The mass action deterministic dynamics associated with $(\S,\C,\Re,\K)$  is 
 \begin{align*}
& a'(t)=-a(t)+b(t)-z(t)(\theta a(t)-\mu),\\
& b'(t)=a(t)-b(t),\\
& z'(t)=z(t)(\theta a(t)-\mu),\\
& c'(t)=-c(t)a(t)+d(t)a(t),\\
& d'(t)=-d(t)a(t)+c(t)a(t),\\
 \end{align*}
Hence the positive steady state $x^*=(a^*,b^*,c^*,d^*,z^*)=(\frac{\mu}{\theta},\frac{\mu}{\theta},\frac{M_2}{2},\frac{M_2}{2},M_1-2\frac{\mu}{\theta})$. This implies that species $A$ and $B$ are ACR species.}

\textcolor{black}{We split $\S$ into $\S_L=\{A,B,C\}$ and $\S_H=\{D,Z\}$. For a given initial condition  $x(0)=(a(0),b(0),c(0),d(0),z(0))$, we have the following reduced network $(\S_L,\C_L,\Re_L,\K_L)$}

\begin{equation*}
  \begin{tikzpicture}[baseline={(current bounding box.center)}, scale=0.8, rounded corners]

     \node[ text=black] (3) at (-4,6)  {\textbf{$0$}};
   \node[ text=black] (4) at (0,6)  {\textbf{$A$}};
   \node[ text=black] (5) at (4,6)  {\textbf{$B$}};
  \node[ text=black] (1) at (0,8)  {\textbf{$C+A$}};

   \path[shorten >=2pt,->,shorten <=2pt]

    (3) edge[bend left, line width=1pt] node {$\mu z(0)$} (4)
    (4) edge[bend left, line width=1pt] node {$\theta z(0)$} (3)
    (4) edge[bend left, line width=1pt] node {$1$} (5)
    (5) edge[bend left, line width=1pt] node {$1$} (4)
     (4) edge[bend left, line width=1pt] node {$d(0)$} (1)
     (1) edge[bend left, line width=1pt] node {$1$} (4)
       ;
    
  \end{tikzpicture}
 \end{equation*}
\textcolor{black}{ With $M_1=2\frac{\mu}{\theta}+z(0)$ and $M_2=2d(0)$, the condition \eqref{eq:condition of mean} holds for species $E$ and $Z$. Hence by Theorem \ref{thm:steady state after freezing}, the positive steady state values of $A$ and $B$ in $(\S_L,\C_L,\Re_L,\K_L)$ are same as the positive steady state values in $(\S,\C,\Re,\K)$, which are $\frac{\mu}{\theta}$ for both species. }\hfill $\triangle$
\end{example}

\subsection{Approximation of controlled stochastic network with hybrid systems}
In this section, we show that the basic ACR controller \eqref{eq:basic acr} can also be used to control a stochastic system under the classical scaling regime. To show that the target species in the scaled stochastic system follows approximately a Poisson distribution, we use a hybrid type system. This framework is introduced in \cite{anderson2017finite}.

We define $\gamma_k=1$ if $q_H(y_k)> 0$ for a reaction $y_k \to y'_k \in \Re$, otherwise $\gamma_k=0$. As shown in Section \ref{subsec:model reduction}, for a given set of intensity functions $\K$, we define new scaled reaction intensities as
\begin{equation}\label{eq:classic scale intensities}
\bar \K^N= \left \{\bar \lambda^N_{k} = \frac{\kappa_{y\to y'}}{N^{\Vert q_H(y) \Vert_1-\gamma_k}}\lambda_k : \lambda_k \in \K \right \}.
\end{equation}

Suppose $\S=\S_L\cup \S_H$ as \eqref{eq:S setting} for a reaction network $(\S,\C,\Re)$. Suppose further that the projected network $(\S_L,\C_L,\Re_L)$ has deficiency of zero and is weakly reversible. According to \cite[Theorem 3.8 and Corollary 4.4]{anderson2017finite}, the system dynamics for $(\S,\C,\Re,\tilde \K^N)$ can be approximately studied with a hybrid system. The copy numbers of the species in $\S_L$ are modeled with a stochastic process, and the distribution of $q_L(X^N)$ is approximated with a product form of Poissons at a finite time $t$. On the other hand, the concentration of the species in $\S_H$ follows the deterministic dynamics.
These two processes are coupled as the rate constants of the stochastic part are determined by the deterministic part, and the kinetics of the deterministic part also depends on the stationary mean of the stochastic part. We denote by $(\S_H,\C_H,\Re_H,\bar \K_H)$ the deterministic part of the hybrid system. More precise definition of $(\S_H,\C_H,\Re_H,\bar \K_H)$ can be found in \cite{anderson2017finite}.

In the scaled stochastic system modeled with $\bar \K^N$ defined at \eqref{eq:classic scale intensities}, let $m^N(t)$ be the mean copy number of the target species $X_1$ at time $t$ such that $\dlim_{N \to \infty}m^N(t)=m(t)$. That is, $m(t)$ approximates the mean of $X_1$ in the scaled stochastic system.
 The following Lemma provides the conditions guaranteeing that $m(t)$ converges in time to the desired value in a controlled system with the basic ACR controller. As we assumed for the basic ACR system in Section \ref{sec:acr controllers}, we suppose that $X_1$ is in a given reaction network and $Z$ is a newly introduced control species.

\begin{lem}\label{lem:when a hybrid has a steady state}
For a given reaction network, let $(\S,\C, \Re)$ be the union of the given network and the basic ACR system \eqref{eq:basic acr}. Let $X^N$ be the stochastic process associated with $( \S, \C, \Re, \bar \K^N)$, where $\bar \K^N$ is a set of scaled reaction intensities defined as \eqref{eq:classic scale intensities}. For an initial state of $X^N(0)$, suppose $X_1 \in \S_L$ and  $Z \in \S_H$. Suppose further that the conditions of Theorem 3.8 in \cite{anderson2017finite} hold. If the concentration of $Z$ in the deterministic part $(\S_H,\C_H,\Re_H,\bar \K_H)$ converges to a positive steady state, as $t\to \infty$, then $\dlim_{t\to \infty}m(t)=\frac{\mu}{\theta}$.
\end{lem}
\begin{proof}
Since $Z$ is newly introduced species with the basic ACR system, $Z$ is regulated with only two reactions $\emptyset \leftarrow Z \to 2A$ in  $(\S_H,\C_H,\Re_H,\bar \K_H)$.
Let $z(t)$ be the concentration of species $Z$ at time $t$. According to Theorem 3.8 in \cite{anderson2017finite}, the rate of reaction $\emptyset \leftarrow Z$ is $\mu z(t)$, and the rate of reaction $Z\to 2A$ is $\theta m(t)z(t)$. Then $z(t)$ solves a differential equation,
\begin{align*}
\frac{d}{dt}z(t)=z(t)\left ( \theta m(t)-\mu \right),
\end{align*}

Hence if $\dlim_{t\to \infty}z(t) =z^*$ for some $z^*\in (0,\infty)$, then $\dlim_{t \to \infty}\left (\theta m(t)-\mu\right)=0$.

\end{proof}
In Section \ref{subsec:dimer-catal}, we use this lemma to show that the mean of the species $X$ in the hybrid system in Figure \eqref{fig:faststochastic}d converges to $\dfrac{\mu}{\theta}$.

\subsection{Foster-Lyapunov criterion}\label{sec:lyapunov}
By equation \eqref{eq:conv}, it is important to show the term $|p(A,t)-\pi(A)|$ is small for the stationary distribution approximation \eqref{eq:poisson approx} with a Poisson distribution.
 Basically the term $|p(A,t)-\pi(A)|$ tends to zero as $t \to \infty$ by the ergodic theorem \cite{NorrisMC97}. In this section we introduce one of the most well-known theoretical frameworks, the so-called Foster-Lyapunov criterion \cite{MT-LyaFosterIII} for the convergence of $|p^N(A,t)-\pi(A)|$ in $t$. The following theorem is a version of Theorem 6.1 in \cite{MT-LyaFosterIII}.

\begin{thm}[Foster-Lyapunov criterion for exponential ergodicity]\label{thm:exp_ergo}
Let $X$ be a continuous-time Markov chain on a countable state space $\mathbb{S}$ with the infinitesimal generator $\mathcal A$ \eqref{gen5}. Suppose there exists a positive function $V$ on $\mathbb{S}$ satisfying the following conditions. 
\begin{enumerate}
\item $V(x) \to \infty$ as $|x|\to \infty$, and
\item There are positive constants $a$ and $b$ such that
\begin{equation}\label{eq:exp_ergo}
\mathcal{A}V(x) \le -aV(x)+b \quad \text{for all} \ \ x \in \mathbb{S}.
\end{equation}
\end{enumerate}
Then $X$ admits a unique stationary distribution $\pi$. Furthermore, there exists positive constants $\eta$ and $C$ such that for any measurable set $A$ and any state $x$,
\begin{align*}
| P(X(t)\in A| X(0)=x) - \pi(A) | \le C(V(x)+1)e^{-\eta t}.
\end{align*}
 	
 	\end{thm}

For a finite time $t$, if the projected network system in Theorem \ref{thm:main1} satisfies the conditions in Theorem \ref{thm:exp_ergo}, then the term $|p(A,t)-\pi(A)|$ in the right hand side of \eqref{eq:conv2} can be small. Therefore $|p^N(A,t)-\pi(A)|$ in \eqref{eq:conv2} can eventually be small enough with sufficiently large $N$.

\section{Application of the stochastic dynamics}
 
\subsection{A receptor-ligand signaling model}
\label{subsec:RLsupplementary}

In this section we provide the initial reaction propensities in the receptor-ligand signaling model described in Figure \ref{fig:RL schematic}a. Note that we model the reaction propensities with stochastic mass-action kinetics \eqref{mass}. Letting $N=10^3$, we set the initial values $L(0)=1.5N, P(0)=100, R(0)=2, Z(0)=N$, and we set zero initial values for the rest of species. The parameters are $\kappa^N_1=1.24/(1.5N), \kappa^N_2=1.37, \kappa^N_3=1.41, \kappa^N_4=1.79,  \kappa^N_5=1.02, \kappa^N_6=1.36, \kappa^N_7=1.97, \kappa^N_8=1.11, \kappa^N_{9}=1.55, \kappa^N_{10}=1.01,\kappa^N_{11}=1.34, \kappa^N_{12}=0.5, \theta=1/N$ and $\mu=5/N$. Then for the initial condition $X(0)=(L(0),R_0(0),R(0),D(0),D_1(0),D_2(0),D_3(0),Z(0))$, we have the following propensities of reaction $L+R_0\to R, R\to L+R_0$ and $2R\to D$.
\begin{align*}
&\lambda_{L+R_0\to R}(X(0))=\kappa_1L(0)R_0(0)=2.46, \quad \lambda_{R\to L+R_0}(X(0))=\kappa_2R(0)=2.74\\
&\lambda_{2R\to D}(X(0))=\kappa_3R(0)(R(0)-1)=2.82.
\end{align*}
The propensities of all remaining reactions are zero since the initial counts of $D,D_1,D_2$ and $D_3$ are zero. As we mentioned in Section \ref{subsec:stochastic control} of the main text, the initial propensities of each reaction is relatively small to the initial counts of species $L$ and $Z$. Hence it needs a longer time to substantially deviate the counts of $L$ and $Z$. Over a short term interval, each of $L$ and $Z$ in the associated stochastic process behaves as a constant function (Figure \ref{fig:RL schematic}c). Hence the reduced model obtained by freezing $L$ and $Z$ at their initial values approximates the original controlled system. 

By using Theorem \ref{thm:main1}, we show this approximation more precisely. We let $(\S,\C,\Re)$ be the controlled system in Figure~\ref{fig:RL schematic}~a. We choose the initial values and the set $\K^N$ of parameters as introduced above with the scaling parameter $N$. Note that $\K^N$ satisfies \eqref{eq:scaled rates}. We split $\S$ into $\S_L=\{R_0, R, D, D_1,D_2,D_3\}$ and $\S_H=\{L,Z\}$. Then the reduced network $(\S_L,\C_L,\Re_L,\K_L)$ is the network in Figure~\ref{fig:RL schematic}~b. The deficiency of the reduced network is $0$ as the number of complexes is $8$, there are two linkage classes, and the rank of the stoichiometry matrix is $6$. The reduced network is also weakly reversible because each linkage class is strongly connected. Hence Corollary \ref{cor:poissonapprox} implies the distribution $p_L^N$ of $S_L$ species at $t=150$ in $(\S,\C,\Re,\K)$ is estimated by a product form of Poissons as described at \eqref{eq:poisson approx}.

In order to approximate the mean of $R_0$ in the controlled system $(\S,\C,\Re)$, we show that the positive steady state value of $R_0$ in $(\S_L,\C_L,\Re_L,\K_L)$ is $\frac{\mu}{\theta}=5$.
Theorem \ref{thm:steady state after freezing} can be used to show that the mean, but using deficiency zero condition of the reduced network provides a much simpler way. Since $(\S_L,\C_L,\Re_L,\K_L)$ has zero deficiency and is weakly reversible, the associated mass action deterministic dynamics is \textit{complex balanced} \cite{Feinberg72, Horn72}. This means that for each complex, all `in-flows` and `out-flows` are balanced at each positive steady state. Hence for the zero complex in $(\S_L,\C_L,\Re_L,\K_L)$, the in-flow is $\theta r_0^*$ and the out-flow is $\mu$ for the positive steady value $r_0^*$ of $R_0$. Therefore they are balanced at $r^*_0=\frac{\mu}{\theta}$.

\tikzset{every node/.style={auto}}
 \tikzset{every state/.style={circle, rounded corners, minimum size=1pt, draw=none, font=\normalsize}}
 \tikzset{bend angle=8}
 We now investigate the accuracy of this approximation. Let $p_L^N(\cdot,t)$ and $p(\cdot,t)$  denote the distribution of species $S_L$ in the controlled receptor-ligand system and the distribution of the reduced system, respectively. We further let $\pi$ be the product form of Poissons stationary distribution of the reduced system.
As shown in \eqref{eq:conv2}, for a small error between $p_L^N(\cdot,t_0)$ and  $\pi$ with a fixed time $t_0=150$, we need a fast convergence for $p(\cdot,t)$ to $\pi$.

With a slight modification on the reduced network, we show that how to use the Foster-Lyapunov criterion in Theorem \ref{thm:exp_ergo} to show the convergence rate of $p(\cdot,t)$ to $\pi$ in time $t$. In order to construct a Lyapunov function explicitly, we add a degradation of $D$ to the reduced model in Figure \ref{fig:RL schematic}b. Hence let $(\S_L,\C_L,\Re_L,\K_L)$ be a system described with the following reaction network.
\begin{equation}\label{eq:reducedforexp}
  \begin{tikzpicture}[baseline={(current bounding box.center)}, scale=0.8, rounded corners]
   \node[text=black ] (9) at (2,6)  {\textbf{$0$}};
    \node[text=black ] (0) at (5,6)  {\textbf{$R_0$}};
     \node[text=black] (1) at (8,6)  {\textbf{$R$}};
   
   \node[ text=black] (3) at (-3,6)  {\textbf{$2R$}};
   \node[ text=black] (4) at (0,6)  {\textbf{$D$}};
   
   \node[ text=black] (6) at (0,4)  {\textbf{$D_1$}};
   \node[ text=black] (7) at (0,2)  {\textbf{$D_2$}};
   \node[ text=black] (8) at (0,0)  {\textbf{$D_3$}};
   
   \path[shorten >=2pt,->,shorten <=2pt]
   (0) edge[bend left, line width=1pt] node {$\kappa_1$} (1)
    (1) edge[bend left, line width=1pt] node {$\kappa_2$} (0)
    (3) edge[bend left, line width=1pt] node {$\kappa_3$} (4)
    (4) edge[bend left, line width=1pt] node {$\kappa_4$} (3)
    
    (4) edge[bend left, line width=1pt] node {$\kappa_5$} (6)
    (6) edge[bend left, line width=1pt] node {$\kappa_6$} (4)
    (6) edge[bend left, line width=1pt] node {$\kappa_7$} (7)
    (7) edge[bend left, line width=1pt] node {$\kappa_8$} (6)
    (7) edge[bend left, line width=1pt] node {$\kappa_{9}$} (8)
    (8) edge[bend left, line width=1pt] node {$\kappa_{10}$} (7)
     (9) edge[bend left, line width=1pt] node {$\mu$} (0)
    (0) edge[bend left, line width=1pt] node {$\theta$} (9)
    (4) edge[line width=1pt] node {$\kappa_{11}$} (9)
       ;
    
  \end{tikzpicture}
 \end{equation}  

Let $x=(x_1,x_2,x_3,x_4,x_5,x_6)$ be a vector each of whose entries represents the copy numbers of $R_0, R, D, D_1, D_2$ and $D_3$, respectively. We use a linear Lyapunov function $V(x)=\sum_{i=1}^6v_ix_i$ with some positive vector $v=(v_1,v_2,\dots,v_d)$. The work in \cite{scalable2014} exploited details about linear Lyapunov functions for stochastic reaction networks. By the definition of the generator $\mathcal{A}$ \eqref{gen5}, we have
\begin{equation}\label{eq:av}
\begin{aligned}
\mathcal{A}V(x)=&\kappa_{10}(v_5-v_6)x_6+(\kappa_9(v_6-v_5)+\kappa_8(v_4-v_5))x_5+
(\kappa_7(v_5-v_4)+\kappa_6(v_3-v_4))x_4\\
&+(\kappa_5(v_4-v_3)+\kappa_4(2v_2-v_3)-\kappa_{11}v_3)x_3++\kappa_2(v_1-v_2)x_2\\
&+(\kappa_1(v_2-v_1)-\theta v_1)x_1
+(\kappa_3(v_3-2v_2))x_2^2+\mu.
\end{aligned}
\end{equation}
Let $h_i=v_{i+1}-v_i$ for $i=1,3,4,5$ and $h_2=v_3-2v_2$. We will choose $h_i$'s with which all the coefficients on the right hand side of \eqref{eq:av} are strictly negative numbers. Hence, we choose $h_i$'s such that 
\begin{equation}\label{eq:ineq}
\begin{aligned}
& h_1<0, \quad  h_2>0,  \quad \frac{\kappa_5}{\kappa_{11}}h_3-\frac{\kappa_4}{\kappa_{11}}h_2<v_3=2v_1+2h_1+h_2 \\
& \frac{\kappa_7}{\kappa_6}h_4< h_3,\quad \frac{\kappa_9}{\kappa_8}h_5< h_4,  \quad \text{and} \quad 
 h_5>0.
\end{aligned}
\end{equation}
With a sufficiently large $v_1$, we can find $h_i$'s satisfying \eqref{eq:ineq}. Then we have
\begin{align*}
\mathcal{A}V(x)=-c_1x_1-c'_2x^2+c_2x_2-c_3x_3-c_4x_4-c_5x_4-c_6x_6,
\end{align*}
for some $c_i >0$ for $i=1,2,\dots,6$ and $c'_2>0$. Hence \eqref{eq:exp_ergo} holds, and for the distribution $p(\cdot,t)$ of \eqref{eq:reducedforexp} and its stationary distribution $\pi$, we have the exponential decay of $|p(A,t)-\pi(A)|$ for any $A \subset Z^6_{\ge 0}$, as $t\to \infty$.

\subsection{A Dimer-Catalyzer Model}\label{subsec:dimer-catal}
Recall that we used the hybrid system shown in Figure \ref{fig:faststochastic}d to approximate the distribution of $X$ in the controlled dimer-catalyzer system that we introduced in  Figure \ref{fig:faststochastic}a. In this section, by using Lemma \ref{lem:when a hybrid has a steady state} we prove that $m(t)$, the mean of $X$ at time $t$ in the hybrid system, converges to $\dfrac{\mu}{\theta}$ under a mild condition, as $t\to \infty$.  

Let $x_1(t), c(t), c_p(t), c_{pp}(t), x^*(t)$ and $z(t)$ be the solution to the deterministic part of the hybrid system shown in Figure \ref{fig:faststochastic}d. Let further $m(t)=\dfrac{\kappa_1 x^*(t)+\mu z(t)}{\kappa_2 c_{pp}(t)+\theta z(t)}$ and $\ell(t)=\dfrac{\kappa_4 c(t)}{\kappa_3 x_1(t)}$ be the mean of $X$ and $X_2$ at time $t$, respectively. Then the deterministic system is governed by the following system of ordinary differential equations,
\begin{equation}\label{eq:ode of the hybrid}
\begin{split}
&\frac{d}{dt}x_1(t)=\kappa_4 c(t)-\kappa_3 x_1(t) \ell(t)=0,\\
&\frac{d}{dt}c(t)=\kappa_3 x_1(t) \ell(t)+\kappa_6 c_p(t)-(\kappa_4+\kappa_5)c(t)=\kappa_6 c_p(t)-\kappa_5 c(t),\\
&\frac{d}{dt}c_p(t)=\kappa_5 c(t)+\kappa_8 c_{pp}(t)-(\kappa_6+\kappa_7)c_p(t),\\
&\frac{d}{dt}c_{pp}(t)=\kappa_7 c_p(t)-\kappa_8 c_{pp}(t),\\
&\frac{d}{dt}x^*(t)=\kappa_2 c_{pp}(t)m(t)-\kappa_1 x^*(t), \quad \text{and}\\
&\frac{d}{dt}z(t)=z(t)(\theta  m(t)-\mu).
\end{split}
\end{equation}

Let $x_1(0), c(0), c_p(0), c_{pp}(0), x^*(0)$ and $z(0)$ be the initial values of the system \eqref{eq:ode of the hybrid}. Note that $\dfrac{d}{dt}x^*(t)+\dfrac{d}{dt}z(t)=0$ and $\dfrac{d}{dt}c(t)+\frac{d}{dt}c_p(t)+\dfrac{d}{dt}c_{pp}(t)=0$ for all $t$. We let $M=x^*(0)+z(0)=x_1(t)+z(t)$ and $L=c(0)+c_{p}(t)=c(t)+c_p(t)+c_{pp}(t)$ be the conserved quantities.

By using the second, third and fourth equations in \eqref{eq:ode of the hybrid}, there is a single positive steady state, $$(\bar{c}, \bar{c}_p, \bar{c}_{pp})=\left (\dfrac{\kappa_6 \kappa_8 L}{\kappa_6\kappa_8+\kappa_5 \kappa_8+\kappa_5\kappa_7}, \dfrac{\kappa_5 \kappa_6 \kappa_8 L}{\kappa^2_6\kappa_8+\kappa_5\kappa_6 \kappa_8+\kappa_5\kappa_6\kappa_7}, \dfrac{\kappa_5\kappa_6 \kappa_7\kappa_8 L}{\kappa^2_6\kappa^2_8+\kappa_5\kappa_6 \kappa^2_8+\kappa_5\kappa_6\kappa_7\kappa_8} \right ),$$ for $x_1(t),c(t),c_p(t)$ and $c_{pp}(t)$. In order to study the stability of this positive steady state, we consider the following linear system for $c, c_p$ and $c_{pp}$, 
\begin{align*}
\frac{d}{dt}\begin{pmatrix}
c\\ c_p\\c_{pp}
\end{pmatrix}
=  \begin{pmatrix}
-\kappa_5& \kappa_6 & 0\\
\kappa_5& -\kappa_6-\kappa_7 & \kappa_8\\
0& \kappa_7 & -\kappa_8
\end{pmatrix}
\begin{pmatrix}
c\\ c_p\\c_{pp} 
\end{pmatrix}.
\end{align*}
The eigenvalues of the matrix above is 
\begin{align*}
&\lambda_1=0,\\
&\lambda_2=-\frac{1}{2}(\kappa_5+\kappa_6+\kappa_7+\kappa_8)-\frac{1}{2}\sqrt{(\kappa_5+\kappa_6-\kappa_7-\kappa_8)^2+4\kappa_6 \kappa_7},\\
&\lambda_2=-\frac{1}{2}(\kappa_5+\kappa_6+\kappa_7+\kappa_8)+\frac{1}{2}\sqrt{(\kappa_5+\kappa_6-\kappa_7-\kappa_8)^2+4\kappa_6 \kappa_7}.
\end{align*}
Since $(\kappa_5+\kappa_6+\kappa_7+\kappa_8)^2 > (\kappa_5+\kappa_6-\kappa_7-\kappa_8)^2+4\kappa_6 \kappa_7$, the eigenvalues $\lambda_2$ and $\lambda_3$ are strictly negative. Thus 
\begin{align*}
\begin{pmatrix}
c(t)\\ c_p(t)\\c_{pp}(t) 
\end{pmatrix}=
v_1+v_2e^{-\lambda_2 t}+v_3 e^{-\lambda_3 t},
\end{align*}
where $v_i$'s are the corresponding eigenvectors. In particular, $v_1=(\bar{c}, \bar{c}_p, \bar{c}_{pp})$. 
This implies that $\dlim_{t\to \infty} c_{pp}(t)=\bar{c}_{pp}$.

Now in the following proposition we introduce a sufficient condition for $z(t)$ to converge to some positive steady state.
\begin{prop}\label{prop:hybrid}
Suppose that $\theta \kappa_1 M> \mu \kappa_2 \bar c_{pp}$. Then $z(t)$ in \eqref{eq:ode of the hybrid} converges to $\dfrac{\theta \kappa_1 M-\mu \kappa_2 \bar c_{pp}}{\theta \kappa_1}$, as $t \to \infty$.
\end{prop}
\begin{proof}
By using the conservative quantity $M=x^*(t)+z(t)$, the differential equation for $z$ in \eqref{eq:ode of the hybrid} can be written as
\begin{align*}
\frac{d}{dt}z(t)=z(t) \left (\frac{\theta \kappa_1 M-\mu \kappa_2 c_{pp}(t)-\theta \kappa_1 z(t)}{\kappa_2 c_{pp}(t)+\theta z(t)} \right )
=\frac{ \theta \kappa_1 z(t)}{\kappa_2 c_{pp}(t)+\theta z(t)}(\alpha(t)-z(t)),
\end{align*}
where $\alpha(t)=M-\dfrac{\mu \kappa_2}{\theta \kappa_1} c_{pp}(t)$.

Let $\bar \alpha= M-\dfrac{\mu \kappa_2}{\theta \kappa_1} \bar c_{pp}$. Let also $\epsilon>0$ be an arbitrarily small number.
Since $\dlim_{t\to \infty}c_{pp}(t)=\bar c_{pp}$, there exists a $T>0$ such that $|\alpha(t)-\bar \alpha|< \epsilon$ when $t>T$. We denote $R_{-\epsilon}=\bar \alpha-2\epsilon$ and $R_{+\epsilon}=\bar \alpha+2\epsilon$. Then clearly $\alpha(t) \in [R_{-\epsilon},R_{+\epsilon}]$ for all $t>T$.

In the rest of the proof, we suppose $t>T$.
Note that if $z(t) < R_{-\epsilon}$, then
\begin{align*}
\dfrac{d}{dt}z(t) = \frac{ \theta \kappa_1 z(t)}{\kappa_2 c_{pp}(t)+\theta z(t)}(\alpha(t)-z(t)) > \frac{ \theta \kappa_1 z(t)}{\kappa_2 c_{pp}(t)+\theta z(t)}(\alpha(t)-\bar \alpha +2\epsilon) > \frac{ \theta \kappa_1 z(t)}{\kappa_2 L+\theta M}\epsilon
\end{align*} 
Note further that if $z(t)> R_{+\epsilon}$, then
\begin{align*}
\dfrac{d}{dt}z(t) =  \frac{ \theta \kappa_1 z(t)}{\kappa_2 c_{pp}(t)+\theta z(t)}(\alpha(t)-z(t)) < \frac{ \theta \kappa_1 z(t)}{\kappa_2 c_{pp}(t)+\theta z(t)}(\alpha(t)-\bar \alpha -2\epsilon) < -\frac{ \theta \kappa_1 z(t)}{\kappa_2 L+\theta M}\epsilon
\end{align*}

Therefore there exists $T_1=\inf \{t\ge T : z(t) \in [R_{-\epsilon},R_{+\epsilon}]\} < \infty$. 
Let $T_2=\inf \{t>T_1 : z(t) =\bar \alpha \}$.
We first consider the case $z(T_1)=R_{-\epsilon}$ and $T_2<\infty$. 
If there exists $t>T_2$ such that $z(t)>\alpha(t)$, then  by using the continuity of $z(t)$, 
\begin{align*}
z(t) = z(T_2) + \int_{T_2}^t \frac{d}{ds}z(s) ds = \alpha(t)+\int_{T_2}^t \frac{ \theta \kappa_1 z(s)}{\kappa_2 c_{pp}(s)+\theta z(s)}(\alpha(s)-z(s))ds < \alpha(t).
\end{align*}
Hence it is contradiction to $z(t)>\alpha(t)$. Thus $  R_{-epsilon}< z(t) < \alpha(t)$ for all $t>T_2$. Now we consider the case $z(T_1)=R_{-\epsilon}$ and $T_2=\infty$. Then  for any $t>T_1$, it follows that
\begin{align*}
z(t) = z(T_1) + \int_{T_1}^t \frac{d}{ds}z(s) ds = R_{-epsilon}+\int_{T_2}^t\frac{ \theta \kappa_1 z(s)}{\kappa_2 c_{pp}(s)+\theta z(s)}(\alpha(s)-z(s)) ds > R_{-\epsilon} 
\end{align*}
as $z(t) < \alpha(t)$ for any $t>T_1$. Hence $R_{-\epsilon} < z(t) < \alpha(t)$ for all $t>T_1$. Hence we conclude that when $z(T_1)=R_{-\epsilon}$, there exists $T'_1\ge T_1 \ge T$ such that such that $\alpha(t) < z(t) < R_{+\epsilon}$ for all $t>T'_1$.  In the same way, it follows that when $z(T_1)=R_{+\epsilon}$ there exists $T'_2\ge T_1 \ge T$ such that $\alpha(t) < z(t) < R_{+\epsilon}$ for all $t>T_3$. 

Consequently, there exists $T'\ge T$ such that $z(t)\in [R_{-\epsilon},R_{+\epsilon}]$ for all $t>T'$. Since $\epsilon$ can be chosen arbitrarily small with sufficiently large $T$, $z(t)$ converges to $\bar \alpha$, as $t\to \infty$. Obviously $\bar \alpha$ is positive by the assumption $\theta \kappa_1 M> \mu \kappa_2 \bar c_{pp}$.

\end{proof}

Proposition \ref{prop:hybrid} shows that for any choice of system parameters $\kappa_i$, if either the initial value of $x^*$ or $z$ is large enough, then $z(t)$ converges to a positive steady state, as $t$ goes to $\infty$. Hence by using Lemma \ref{lem:when a hybrid has a steady state}, we conclude that $\dlim_{t\to \infty} m(t)=\dfrac{\mu}{\theta}$ if we input sufficiently large initial concentration of the control species $Z$.

\bibliographystyle{unsrtnat}

\end{document}